\documentclass[aps,prx,twocolumn, notitlepage,longbibliography,superscriptaddress,nofootinbib]{revtex4-2}
\usepackage{amsmath}
\usepackage{amsthm}
\usepackage{thm-restate}
\usepackage{amsfonts}
\usepackage{amssymb}
\usepackage{dsfont}
\usepackage{graphicx}
\usepackage{physics}
\usepackage{tikz}
\usepackage{hyperref}
\usepackage{mathtools}
\usepackage{comment}
\usepackage{multirow}
\usepackage{ulem}
\usepackage{pst-node}
\usepackage{tikz-cd}
\usepackage[shortlabels]{enumitem}
\usepackage{mathtools}

\usepackage[english]{babel}

\newtheorem{theorem}{Theorem}
\newtheorem{corollary}{Corollary}[theorem]
\newtheorem{definition}{Definition}
\newtheorem{lemma}{Lemma}

\usepackage{pst-all}
\usepackage{leftidx}

\def\M{\mathcal{M}}
\def\Z{\mathbb{Z}}
\def\C{\mathcal{C}}
\def\D{\mathcal{D}}

\def\L{\mathcal{L}}

\DeclareMathOperator*{\Motimes}{\text{\raisebox{0.25ex}{\scalebox{0.8}{$\bigotimes$}}}}

\newcommand{\ZZ}{\mathbb{Z}}

\newcommand{\be}{\begin{equation}}
\newcommand{\ee}{\end{equation}}
\newcommand{\bc}{\begin{center}}
\newcommand{\ec}{\end{center}}
\newcommand{\nin}{\noindent}

\newcommand{\non}{\nonumber}
\newcommand{\lo}{\overline}

\newcommand{\as}{\mathsf{a}}
\newcommand{\bs}{\mathsf{b}}
\newcommand{\cs}{\mathsf{c}}

\definecolor{dualblue}{RGB}{3,101,192}

\usetikzlibrary{arrows.meta}
\makeatletter
\newlength\lrvec@height
\newlength\lrvec@width
\newif\iflrvec@same@height
\def\lrvec{\@ifstar\slrvec@\lrvec@}
\newcommand{\slrvec@}[2][.4ex]{
  \lrvec@same@heighttrue
  \mathpalette\lrvec@@{{#1}{#2}}
}
\newcommand{\lrvec@}[2][.4ex]{
  \lrvec@same@heightfalse
  \mathpalette\lrvec@@{{#1}{#2}}
}
\def\lrvec@@#1#2{\lrvec@@@#1#2}
\def\lrvec@@@#1#2#3{%
  \iflrvec@same@height
    \settoheight{\lrvec@height}{$\m@th#1 \mathbf{T}#3$}
  \else
    \settoheight{\lrvec@height}{$\m@th#1#3$}
  \fi
  \settowidth{\lrvec@width}{$\m@th#1#3$}
  \kern.08em
  \raisebox{#2}{\raisebox{\lrvec@height}{\rlap{%
    \kern-.05em
    \begin{tikzpicture}[<-> /.tip={To[width=.4em, length=.2em]}]
      \draw [<->] (-.05em,0)--(\lrvec@width+.05em,0);
    \end{tikzpicture}%
  }}}%
  #3
  \kern.08em
}
\makeatother

\usepackage[bottom]{footmisc}

\begin{document}

\title{Transversal non-Clifford gates on qLDPC codes breaking~the~$\sqrt{N}$~distance barrier and quantum-inspired geometry with $\ZZ_2$ systolic freedom}

\author{Guanyu Zhu}
\affiliation{IBM Quantum, T.J. Watson Research Center, Yorktown Heights, NY 10598 USA}

\begin{abstract}
Historically,  a  $\sqrt{N}log^{1/2}(N)$ distance barrier for quantum low-density parity-check (LDPC) codes with 
$N$ qubits persisted for nearly two decades, until the recent discovery of the fibre-bundle code \cite{fiberbundlecode21}. An open question is whether such a distance barrier can be broken while preserving the ability to perform transversal non-Clifford gates.  In this direction, another long-standing distance barrier of $N^{1/3}$ for LDPC stabilizer codes---present since the discovery of the 3D color code---was only recently overcome by a construction achieving an $\Omega(\sqrt{N})$ distance \cite{zhu2025topological}.  The present work further breaks the   $\sqrt{N}$ distance barrier by taking a homological product of three good qLDPC codes, combined with the Freedman–Hastings code-to-manifold mapping  and the triple cup product to implement transversal CCZ gates. The resulting code achieves an  $\Omega(N^{2/3})$  distance (a linear  $X$-distance of $\Theta(N)$) and a dimension of $\Theta(N^{2/3})$, which enables fault-tolerant preparation of $\Theta(N^{1/3})$ independent logical CCZ magic states in a single shot, without distillation (`\textit{magic state fountain}'). This new quantum code also inspires the discovery of a family of exotic $3q$-dimensional manifolds $\mathcal{M}$, which exhibit both a power-law $\mathbb{Z}_2$-($q$, $2q$)-systolic freedom and  $\Theta(vol(\mathcal{M}))$ triple intersection points of $2q$-dimensional submanifolds.

\end{abstract}.

\maketitle

\tableofcontents

\section{Introduction}

Quantum information science has entered a golden age, marked by rapid advances toward building large-scale fault-tolerant quantum computers in recent years. 
A fundamental question is: what is the minimal space-time complexity necessary to perform universal fault-tolerant computation?   Interestingly, this complexity is deeply intertwined with mathematical structures from geometry and topology.

Almost two decades ago, Freedman, Meyer and Luo observed the deep connection between the space overhead of a quantum code and the systolic geometry of the underlying manifold.  In their seminal work \cite{Freedman_systole_2002},   they have constructed the first qLDPC code with distance exceeding the $\sqrt{N}$ distance barrier by a $log^{1/2}(N)$ factor and show the connection to the  $\ZZ_2$-systolic freedom of a manifold, a notion first proposed by Gromov \cite{Gromov:1996}.  Their record was held for nearly two decades until the recent discovery of the Fibre-bundle code by Hastings and Haah \cite{fiberbundlecode21} broke the distance barrier by a power-law factor.   This fiber bundle idea was further developed within the framework of the lift product code \cite{PK:almost_linear} and the balanced product code \cite{Breuckmann:2021_balanced}, until the final achievement of the asymptotically good qLDPC code by Panteleev and Kalachev \cite{pkldpc22}.

Although the quantum memory with optimal space overhead has been achieved by the good qLDPC code, it remains an open question whether there exists a `\textit{good quantum processor}' that achieves constant space-time overhead for universal fault-tolerant computation along with a linear distance.  

In recent years, there has been significant progress for both Clifford and non-Clifford logical gates in qLDPC codes \cite{cohen22, huang2023homomorphic, zhu2023non, xu2024fast, cross2024improved, williamson2024low, swaroop2024universal, golowich2024quantum, lin2024transversal, zhu2025topological}. 
Nonetheless, ever since the discovery of the transversal $T$ gate in 3D color code by Bombin in more than a decade ago  \cite{Bombin:2015jk} (see also Refs.~\cite{Kubica:2015br, Vasmer2019}), there had existed an $N^{\frac{1}{3}}$ distance barrier for transversal non-Clifford gates in topological stabilizer codes and more generally LDPC stabilizer codes (more generally also for constant-depth local circuit) for a long time.  The $N^{\frac{1}{3}}$ distance barrier for the (conventional) topological stabilizer codes, i.e., those defined on the cellulation of an Euclidean geometry, is implied by the Bravyi-Konig bound \cite{Bravyi:2013dx} \footnote{We note that there is no rigorous proof on that yet.}, which states that in order to get a logical gate at the 3rd level of Clifford hierarchy, one needs to define the code on the cellulation of a manifold of dimension
at least 3. For instance, in a 3D torus (with zero curvature), the distance is determined by the minimal length of the logical string operator, which scales as $O(N^{\frac{1}{3}})$.  However, such a constraint is only a consequence of the Euclidean geometry. On the other hand, as we will see, a large class of qLDPC codes can be viewed as non-Euclidean geometries.  More specifically, the fibre bundle idea in Refs.~\cite{fiberbundlecode21, PK:almost_linear, Breuckmann:2021_balanced,  pkldpc22} which implements a twist in the product construction can be applied to break the $N^{\frac{1}{3}}$ distance barrier, as shown very recently in Ref.~\cite{zhu2025topological} which achieves an $\Omega(\sqrt{N})$ distance.

One key insight we need here is that one can build manifolds from quantum codes, as recently shown by Freedman and Hastings \cite{freedman:2020_manifold_from_code}. In particular, they showed that a large class of qLDPC codes based on general chain complexes and expanders, including all the recent fibre-bundle-based constructions in Refs.~\cite{fiberbundlecode21, PK:almost_linear, Breuckmann:2021_balanced,  pkldpc22}, can be mapped to a high-dimensional manifold (with the minimal dimension 11) with bounded local geometry.
This breakthrough erodes the boundary between qLDPC codes defined on general chain complexes and homological qLDPC  codes defined on the cellulation of manifolds.  It also shows an entirely new way of discovering exotic geometries from the combinatorial construction of error correcting codes.
Indeed, it is through this mapping  with the input of the recent codes with distance larger than $\sqrt{N}$ \cite{fiberbundlecode21, PK:almost_linear, Breuckmann:2021_balanced,  pkldpc22} that they are able to build the first manifold with a power-law systolic freedom \cite{freedman:2020_manifold_from_code}. 

To date, there have been mainly two independent threads of exploration of constructing transversal non-Clifford gates on qLDPC codes.  The first is a geometric approach initiated in Ref.~\cite{zhu2023non} which constructed the first high-rate qLDPC code with non-Clifford logical CCZ gates, the 3D quasi-hyperbolic code with an almost-linear dimension $K=\Theta(N/\log N)$ and a distance growing as $d=\Omega(\log N)$. It is also the first application of cup products to high-rate qLDPC codes, and the connection between logical gates and cup products can be traced back to Ref.~\cite{barkeshli2023codimension, barkeshli2024higher}, while the essential idea appeared even earlier from the connection to emergent symmetries of topological order \cite{Yoshida_gate_SPT_2015, Yoshida_global_symmetry_2016, Yoshida2017387, Zhu:2017tr, zhu:2022fractal}. The dimension was later improved to linear, i.e.~$K=\Theta(N)$, by combining 
the quasi-hyperbolic code with the quantum rainbow code \cite{scruby2024quantum}. Later on, by using the mapping from classical and quantum codes to triangulated manifolds following the treatment in Ref.~\cite{freedman:2020_manifold_from_code}, one further reaches a qLDPC code with constant stabilizer weight $w=O(1)$, linear dimension $K=\Theta(N)$ and distance $d=\Omega(\sqrt{N})$ \cite{zhu2025topological}, which breaks the $N^{\frac{1}{3}}$ distance barrier.
 Moreover, the triangulation allows one to define triple cup products, which can be used to construct transversal logical CCZ gates.  Geometrically, the triple cup product of three cocycles summed over the entire manifold counts the $\ZZ_2$ triple intersection number of their Poincar\'e dual cycles. Whenever three logical cycles get a non-trivial triple intersection, the corresponding three logical qubits are acted by a logical CCZ.  Therefore, the number of $\ZZ_2$ triple intersection points corresponds to the number of logical CCZ's, and there are $\Theta(N)$ triple intersection points in the construction from   Ref.~\cite{zhu2025topological}.   

The other independent line of exploration is through an algebraic approach, which uses the hypergraph-product of three algebraic codes based on the classical Reed-Muller codes \cite{golowich2024quantum, lin2024transversal}.  The construction achieves an almost-linear dimension $K=\Theta(N^{1-\epsilon})$, power-law distance $d=\Omega(N^{\frac{1}{3}}/polylog(N))$ and has a quasi-LDPC property with a stabilizer weight of $\Theta(polylog(N))$. The origin of the $polylog(N)$ stabilizer weight and the power-law reduction of the linear dimension is due to the conversion from large-dimensional $\mathbb{F}_q$ qudit code (with growing dimension $q$) to $\ZZ_2$ qubit code  and the intrinsic code parameter constraint in the Reed-Muller code.  This approach also uses cup products (based on $\mathbb{F}_q$-homology instead of $\ZZ_2$-homology) to construct logical gates as in Ref.~\cite{zhu2023non, zhu2025topological}, and generalizes it to quantum sheaf codes \cite{Meshulam:2018}, i.e., generalized Sipser-Spielman codes defined on simplicial or cubical complexes  with local codes.    There is also another work which generalizes the cup product approach to a certain type of more general chain complexes satisfying some combinatorial conditions \cite{breuckmann2024cups}.    

The main difference between the geometric approach in Ref.~\cite{zhu2025topological} and the algebraic approach in Refs.\cite{golowich2024quantum, lin2024transversal} (also  Ref.~\cite{breuckmann2024cups}) is that the latter requires the underlying classical codes to satisfy certain combinatorial conditions, such as the multiplication properties of local codes in Refs.\cite{golowich2024quantum, lin2024transversal} in order to have well-defined cup product on a sheaf or certain type of more general chain complexes.  On the other hand, the geometric approach in  Ref.~\cite{zhu2025topological} is fully topological and insensitive to local combinatorial details. It can take the Tanner graphs of any classical codes or a large classes of quantum codes based on product construction as skeletons to build a manifold with a triangulation  such that the cup products are well-defined.  These triangulations are essentially high-dimensional simplicial complexes, which were also used as alternatives to build asymptotically good classical codes by placing the bits on higher-dimensional simplices and  without using local codes \cite{10.13069/jacodesmath.617235} (based on the random complex constructions in Refs.~\cite{10.1007/s00454-017-9926-3, 10.48550/arxiv.1010.1400, 10.1007/s00493-006-0027-9, 10.48550/arxiv.math/0609773, 10.48550/arxiv.math/0609773}).  This flexibility of choosing input codes, including the good qLDPC codes, allows Ref.~\cite{zhu2025topological} a larger distance $\Omega(\sqrt{N})$ and a constant rate, as well as a constant stabilizer weight (LDPC condition) which is important for the code to have a threshold under circuit-level noise.

Now in order to further break the $\sqrt{N}$ distance barrier, one may need to obtain a geometry with a power-law $\ZZ_2$-systolic freedom, and in addition, one also needs the coexistence of non-trivial triple intersections of large cycles in such a manifold.  This is indeed achieved in the present work.  In the following subsection, we summarize the main results of this paper.

\subsection{Summary of results}

As a preliminary, we first review the cup product formalism  introduced to QEC in Ref.~\cite{zhu2023non} and further developed in Refs.~\cite{zhu2025topological, Hsin2024:classifying, Hsin2024_non-Abelian}  in Sec.~\ref{sec:cup_product}.  In particular, an operator-valued cochain formalism shows how a constant-depth circuit composed of a product of physical CCZ gates can be written in the form of triple cup product of cochains.  Note that throughout this paper, we use a more relaxed definition of transversal gates which allows overlap of the physical gates in the product and is hence equivalent to a constant-depth local circuit. Such a relaxed definition was also used in recent works in Refs.~\cite{golowich2024quantum, lin2024transversal}, and we use the same convention to ease the comparison with other recent results. The constant-depth circuit then gives rise to the logical CCZ applied on a triple of logical qubits whose logical cycles have a non-trivial triple intersection,  as stated and proved in Lemma \ref{lemma_gate_2}.  Note that the proof in this paper is more specifically based on the Calderbank–Shor–Steane (CSS) codes \cite{nielsen_chuang_2010} to facilitate the understanding for computer scientists, which is different from the more general stabilizer  approach in Refs.~\cite{zhu2023non, zhu2025topological, Hsin2024:classifying, Hsin2024_non-Abelian} applicable to more general stabilizer codes including non-Pauli stabilizer models.  In Sec.~\ref{sec:systole}, we review the notions in systolic geometry and the connections to homological quantum codes.

We then introduce our first code construction in Sec.~\ref{sec:subsystem}: the  `\textit{triple good subsystem code}' defined on the triangulation of a 33-manifold $\M^{33}$ as the homological product of three 11-manifolds ($\M^{33}=\M^{11} \times \M'^{11} \times \M''^{11}$), where the qubits are placed on 11-simplices. These 11-manifolds are built from the asymptotically good qLDPC codes constructed in Ref.~\cite{pkldpc22} via the FH mapping (Theorem \ref{theorem:FH}), which has also been used to construct the 3D local code in Ref.~\cite{portnoy2023local}.   These codes are in fact CSS stabilizer codes with bounded-weight stabilizers.  However, there exist  shorter 11-systole of size $\Theta(N^{1/3})$ in these manifolds. To resolve this issue, we can view the code as a subsystem code by treating the logical qubits associated with short cycles as gauge qubits (Lemma \ref{lemma:subsystem}).  In this way, we can retain a code with subsystem-code distance as $d= \Omega(N^{2/3})$, linear $X$-distance $d_X= \Omega(N)$, and the code dimension $K= \Theta(N^{2/3})$. The linear $X$-distance property will be particular useful for systems with biased noise. The code admits $\Theta(N)$ logical CCZ gates due to the presence of $\Theta(N)$ triple intersection points for triples of $22$-cycles from a given homology basis (Theorem \ref{theorem:subsystem}).  We note that although viewed as a subsystem code, the code has only bounded stabilizer weight $w=O(1)$ satisfying the LDPC conditions in contrast to certain typical subsystem codes with growing stabilizer weight such as the Bacon-Shor code \cite{Bacon:2006hw}.  Therefore, these codes are expected to have an error threshold even in the presence of circuit-level noise since the syndrome measurement has a bounded circuit depth. 

Although for practical purpose, this subsystem code construction is good enough, we would still like to understand conceptually whether the more standard CSS subspace code \footnote{`Subspace code'  here means the code space is a subspace of the entire Hilbert space of $N$ qubits.} can break the $\sqrt{N}$ distance barrier.  Geometrically, that will require the underlying manifold having $\ZZ_2$ systolic freedom, which is not the case for the subsystem construction. We resolve these two questions in Sec.~\ref{sec:subspace}, where we have constructed a family of $3q$-dimensional manifolds (with $q\ge 31$) from the product of three manifolds with generally different dimensions, which has $q$-systole and $2q$-systole scaled as $\Omega(N^{2/3})$ and hence achieves a power-law $\ZZ_2$-$(q, 2q)$-systolic freedom (Theorem \ref{theorem:3q-manifold}). Moreover, such manifolds admit $\Theta(N)$ triple intersection points of three $q$-cycles in a given homology basis, which was absent in the first manifold construction with power-law $\ZZ_2$ systolic freedom in Ref.~\cite{freedman:2020_manifold_from_code}. This is because when considering the coarse geometry at large scale, the 11-manifold in Ref.~\cite{freedman:2020_manifold_from_code} is coarsely 2D according to Ref.~\cite{portnoy2023local}, namely it can be embedded into a 2D non-Euclidean space \footnote{The notion of coarse dimension was first introduced by Gromov.}, which hence does not admit triple intersection of non-trivial cycles with large systoles. On the other hand, the 33-manifold and all the $3q$-manifold constructed in this paper are the triple product of the coarsely 2D manifolds, which are hence coarsely 6D, although they have different topological dimensions at small scale. When defining the subspace CSS code on its triangulation by placing the qubits on the $q$-simplices, we obtain a qLDPC code (`\textit{triple good code}') with distance $d=\Omega(N^{2/3})$ and dimension $K= \Theta(N^{2/3})$ that admits $\Theta(N)$ logical CCZ's (Corollary \ref{corollary:subspace}).

Finally, we analyze the detailed logical gate structure from the triple intersection structure of the underlying manifolds in Sec.~\ref{sec:fountain}, which are isomorphic for both the subsystem and subspace code constructions. Moreover, we use these codes for the application of fault-tolerantly prepare logical magic states in a single shot without distillation, dubbed `\textit{magic state fountain}' in Ref.~\cite{zhu2025topological, zhu2023non}.  The fountain can prepare $\Theta(N^{1/3})$ logical CCZ magic states with distance $\Theta(N^{2/3})$.

\section{Preliminaries}

\subsection{Transversal non-Clifford gates on homological codes via triple cup products}
\label{sec:cup_product}

We first briefly review the formalism of performing logical gates via cohomology operations introduced in Ref.\cite{zhu2023non, zhu2025topological, Hsin2024:classifying}.  We consider a CSS code \cite{nielsen_chuang_2010} defined on an $n$-dimensional simplicial complex $\L$  with qubits placed on the $q$-simplices ($1 \le q \le n-1$), i.e., associated with the $\ZZ_2$ $q$-chain group $C_q$.   The $\ZZ_2$ chain group can be viewed as a $\ZZ_2$ vector space (a free $\ZZ_2$ module). The element of the $q$-chain group is the $q$-chain $c_q \in C_q$, which is a finite linear combination of $q$-simplices with coefficients in $\ZZ_2$. The $q$-chain as a $\ZZ_2$ vector can be expanded as: 
$c_q= \sum_{s_q} c_q(s_q)  s_q$,
where $s_q$ represents a $q$-simplex which can be viewed as a basis vector  and $c_q(s_q) \in \{0,1\}$ is the corresponding $\ZZ_2$ coefficient [see Fig.~\ref{fig:triangulation}(a) for illustration].  
We then place $X$- and $Z$-stabilizers on the $(q-1)$-simplices and $(q+1)$-simplices respectively.   The corresponding chain complex is as follows: 
\begin{align}
&C_n \rightarrow \cdots \rightarrow C_{q+1} \xrightarrow[]{\partial_{q+1}=\mathbf{H}_Z^T} C_q \xrightarrow[]{\partial_{q}=\mathbf{H}_X} C_{q-1} \rightarrow \cdots, \cr
&\qquad   \quad \quad \ Z\text{-stabilizer}  \qquad \   
 \text{qubit} \qquad X\text{-stabilizer}
\end{align}   
where $\partial_q:C_q \rightarrow C_{q-1}$ represents the boundary map [see Fig.~\ref{fig:triangulation}(b)], and $\mathbf{H}_Z$ and $\mathbf{H}_X$ are the parity check matrices  associated with the $Z$- and $X$-stabilizers respectively.     

Now we also introduce the dual description with a $\ZZ_2$ cochain group $C^q$.  A $\ZZ_2$-valued $q$-cochain $c^q \in C^q$ is a function from the set of $q$-simplices to $\ZZ_2$ [see Fig.~\ref{fig:triangulation}(c) for illustration].   The cochain group $C^q$ can be considered as the dual $\ZZ_2$ vector space of that of $C_q$.  One can also identify the $\ZZ_2$ chain and cochain groups, i.e., $C^q = C_q$.   The $q$-cochain as a $\ZZ_2$ vector and can be expanded as $c^q $$= $$ \sum_{s_q} c^q(s_q)  \tilde{s}^q$. Here, $\tilde{s}^q$ is an indicator $q$-cochain that takes value $1$ at the $q$-simplex $s_q$ and $0$ otherwise, which can also be viewed as a basis vector, and the  $\ZZ_2$ coefficient $c^q(s_q) \in \{0,1\}$ is the value of the $q$-cochain on simplex $s_q$.   We can then obtain the cochain complex as a dual description of the same code: 
\begin{align}
&C^n \leftarrow \cdots \leftarrow C^{q+1} \xleftarrow[]{d^{q}=\mathbf{H}_Z} C^q \xleftarrow[]{d^{q-1}=\mathbf{H}_X^T} C^{q-1} \leftarrow \cdots, \cr
&\qquad   \quad \quad \ Z\text{-stabilizer}  \qquad \   
 \text{qubit} \qquad X\text{-stabilizer}
\end{align}  
where $d^q:C^q \rightarrow C^{q+1}$ represents the co-boundary map [see Fig.~\ref{fig:triangulation}(d) for illustration].

\begin{figure}[t]
\includegraphics[width=1\columnwidth]{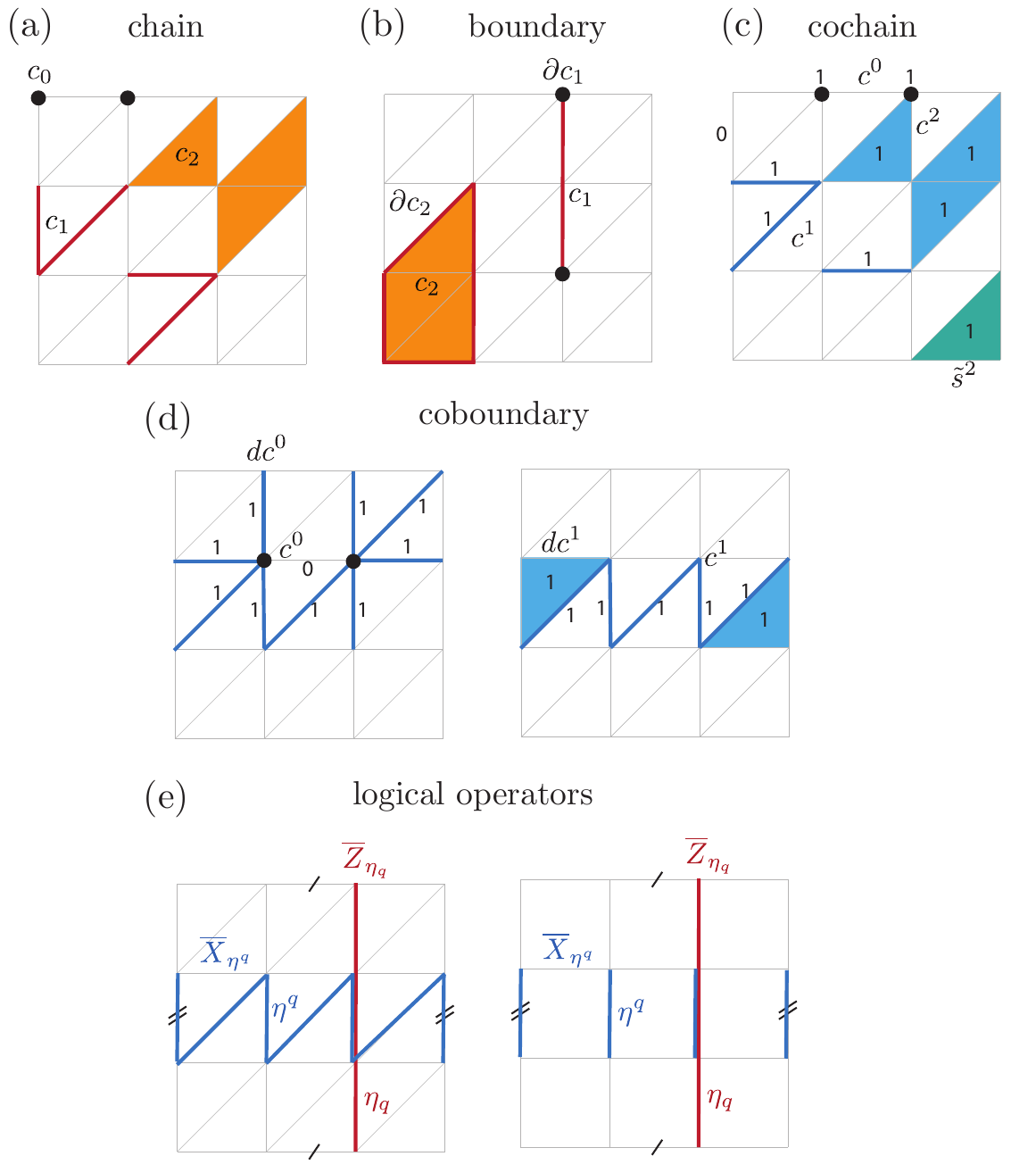}
\caption{(a) Illustration of $\ZZ_2$ $q$-chains as a linear combination of $q$-simplices, including a 0-chain $c_0$, a 1-chain $c_1$ and a 2-chain $c_2$. The vertices (0-simplices), edges (1-simplices) and triangles (2-simplices) with non-zero coefficients are highlighted respectively. (b) Illustration of the boundary map $\partial$ acting on a 1-chain $c_1$ and a 2-chain $c_2$ respectively.  Here, $\partial c_1$ are sum of the end points of the occupied edges in $c_1$, and $\partial c_2$ are sum of the edges surrounding the occupied triangles in $c_2$. (c)  Illustrations of $q$-cochain. The highlighted 0-, 1-, and 2-simplices take value 1 for $c^0$, $c^1$ and $c^2$ respectively, while other simplices take value 0. The green triangles illustrates an indicator 2-cochain $\tilde{s}^2$ which has value 1 on the triangle $s_2$ and value 0 otherwise. (d) Illustration of the coboundary map $d$ acting on a 0-chain $c^0$ (left) and a 1-chain $c^1$ (right).  The highlighted edges and triangles take value 1 for $dc^0$ and $dc^1$ respectively. (e) Illustration of the logical operators $\lo{Z}_{\eta_q}$ $(q=1)$ and its conjugate logical-$X$ operator $\lo{X}_{\eta^q}$ which overlap  on a single edge on a simplicial complex (left) and a square complex (right).  We have imposed periodic boundary conditions so that the simplicial (square) complex serves as the triangulation (cellulation) of a torus.  
 }\label{fig:triangulation}
\end{figure}

Following the convention in the literature, we call the above quantum  code a \textit{$(q,n-q)$-homological code} where $q$ and $n-q$ is the dimension of the cycles where the logical-$Z$ and -$X$ operators are supported respectively. In physics, this type of code also corresponds to a higher-form ($q$-form) $\ZZ_2$ gauge theory.  We now introduce the operator-valued  $\ZZ_2$ $q$-cochains $\hat{a}^q$  with the coefficient being an operator with eigenvalues in $\{0,1\}$, which can be physically interpreted as $q$-form electric gauge fields and the hat $\hat{\cdot}$ indicates that they are operators and hence quantum variables.   We can also expand the operator-valued cochain as $\hat{a}^q $$= $$ \sum_{s_q} \hat{a}(s_q)  \tilde{s}^q$.  The coefficient $\hat{a}(s_q)$ corresponds to a Pauli-$Z$ operator as 
\be\label{eq:a_Pauli}
(-1)^{\hat{a}^q (s_q)}= Z(s_q),
\ee
which has eigenvalues in $\{-1,1\}$.

 Now we can also write any eigenstate in the $Z$-basis as a cochain eigenstate, which can be in turn created from the all-zero state by applying Pauli-$X$ supported on the cochain $c^q$:
\be\label{eq:cochain_eigen}
\ket{c^q} = \prod_{s_q \in c^q} X(s_q)  \ket{00 \cdots 0},
\ee
where $s_q \in c^q$ means $c^q(s_q)=1$. We can hence have an alternative definition of the operator-valued cochains via their eigenstates $\ket{c^q}$:
\be\label{eq:cochain_egenvalues}
\hat{a}^q \ket{c^q} = c^q \ket{c^q},
\ee
where the $q$-cochain $c^q$ as a classical variable  stores the eignvalues of $\hat{a}^q$.  We can also equivalently write down the operator form of $\hat{a}^q$ as
\be
\hat{a}^q :=  c^q \ketbra{c^q}. 
\ee

The logical-$Z$ operator can be defined via the sum of the operator-valued q-cochain $\hat{a}^q$ along a $q$-cycle $\eta_q \in H_q(\L;\ZZ_2)$, i.e., a $q$-cochain satisfying $\partial_q \eta_q $$=$$0$ [$H_q(\L;\ZZ_2) = \text{Ker}(\partial_q)/\text{Img}(\partial_{q+1})$ represents the $q$-th $\ZZ_2$-homology group]: 
\be
\lo{Z}_{\eta_{q}} =(-1)^{\int_{\eta_{q}} \hat{a}^q}=\prod_{s_q \in \eta_{q}} Z(s_q). 
\ee
Here, the discrete sum ${\int_{\eta_{q}} \hat{a}^q} \equiv \sum_{s_q \in \eta_{q}} \hat{a}^q{(s_q)}$ can be interpreted as a chain-cochain paring \footnote{We note that $\int$ here does not represent an integral but a discrete sum.}, i.e., the inner product of the vector associated with the $q$-chain $\eta_q$ and the dual vector associated with the q-cochain $\hat{a}^q$, which equals to the sum of $\hat{a}^q$ over all the $q$-simplices $s_q$ that has the coefficient with value $\eta (s_q)=1$ [see Fig.~\ref{fig:triangulation}(e) for illustration].    

The $Z$-distance of the code counts the smallest weight of any logical-$Z$ operator representative, and is defined to be
\be\label{eq:Z-distance}
d_Z = \min \{ |\eta_q|: \eta_q \neq 0 \in H_q(\L;\ZZ_2)\},
\ee
where $|\eta_q|$ counts the number of $q$-simplices $s_q$ supported on $\eta_q=\sum_{s_q} c(s_q) s_q$ with non-zero coefficient $c(s_q)=1$, i.e., 
  \be\label{eq:Hamming_weight}
   |\eta_q|= \bigg|\sum_{s_q} c(s_q) s_q\bigg|=\sum_{s_q} |c(s_q)|,
  \ee
  which can also be interpreted as the hamming weight since we are dealing with $\ZZ_2$ coefficient.

Similarly we can introduce the operator-valued $\ZZ_2$ $q$-chains $\hat{b}_q$ with eigenvalues in $\{0,1\}$, which can be physically interpreted as $q$-form magnetic gauge fields.  The coefficient of each $q$-simplex $s_q$ corresponds to the Pauli-$X$ operator as
\be\label{eq:b_Pauli}
(-1)^{\hat{b}_q (s_q)}= X(s_q),
\ee
which has eigenvalues in $\{-1,1\}$.
The logical-$X$ operator can be defined via the sum of the operator-valued $q$-chain $\hat{b}_q$ along a $q$-cocycle $\eta^q$ which is a $q$-cochain satisfying $d \eta_q $$=$$0$ ($d$ represents the co-boundary), namely  
\be
\lo{X}_{\eta^{q}} =(-1)^{\int_{\eta^{q}} \hat{b}_q}=\prod_{s_q \in \eta_{q}} X(s_q), 
\ee
[see Fig.~\ref{fig:triangulation}(e)].
Here, $\int_{\eta^{q}} \hat{b}_q \equiv \sum_{s_q \in \eta^{q}} \hat{b}_q{(s_q)} $ is again a chain-cochain pairing,  where $s_q \in \eta^q$ means $\eta^q(s_q)=1$.

The $X$-distance of the code counts the smallest weight of any logical-$X$ operator representative, and is defined to be
\be\label{eq:X-distance}
d_X = \min \{ |\eta^q|: \eta_q \neq 0 \in H_q(\L;\ZZ_2)\},
\ee
where the Hamming weight $|\eta^q|$ counts the number of $q$-simplices $s_q$ supported on $\eta^q=\sum_{s_q} \eta^q(s_q) \tilde{s}_q$ with non-zero value $\eta^q(s_q)=1$, i.e., 
  \be\label{eq:cocycle_Hamming_weight}
   |\eta^q|= \bigg|\sum_{s_q} \eta^q(s_q) \tilde{s}_q\bigg|=\sum_{s_q} |\eta^q(s_q)|.
  \ee
The overall distance for the CSS code is hence 
\be
d=\min\{d_Z, d_X\}.
\ee

Since the quantum code we consider is a CSS code, the associated  code space is
\be\label{eq:CSS_state}
\C := \text{Span} \bigg\{ \overline{\ket{\eta^{q}}} = \frac{1}{|B^{q}|} \sum_{\xi \in B^{q}} \ket{\eta^{q} + \xi} \bigg| \ \eta^{q} \in H^{q}(\L; \Z_2)  \bigg\}.
\ee
Here, $\eta^{q}\in H^{q}(\L; \Z_2)$ represents the $q$-cocycle (classical variable) and can be further expanded in a cocycle basis $\{\alpha^q\}$ as 
\be
\eta^{q} = \sum_{\alpha^q} n_\alpha \alpha^q,
\ee
where $n_\alpha \in \{0,1\}$ is the winding number for each basis $q$-cocycle $\alpha^q$.   Now $\overline{\ket{\eta^{q}}} \equiv \Motimes_{\alpha}\overline{\ket{ n_\alpha  }}$  is the logical-$Z$ eigenstate with a $+1$ eigenvalue for a logical-$Z$ operator $\lo{Z}_{\eta_{q}}$ acted  along the conjugate $q$-cycle $\eta_{q}$, i.e.,      
\be
\lo{Z}_{\eta_{q}} \overline{\ket{\eta^{q}}} \equiv  (-1)^{\int_{\eta_{q}} \hat{a}^q} \ \overline{\ket{\eta^{q}}}= \overline{\ket{\eta^{q}}}. 
\ee
In Eq.~\eqref{eq:CSS_state}, $\ket{\eta^q}$ represents a classical codeword state in the $Z$-basis: 
\be\label{eq:codeword_state}
\ket{\eta^q} = \lo{X}_{\eta^q}\ket{00\cdots 0} = \prod_{s_q \in \eta^q}X(s_q) \ket{00\cdots 0},
\ee
which is a specific case of the cochain eigenstate $\ket{c^q}$ in Eq.~\eqref{eq:cochain_eigen}.
Note that $\ket{\eta^q}$ can also be considered as a cocycle eigenstate, since it is the eigenstate of the operator-valued cocycle $\hat{a}^q$  according to Eq.~\eqref{eq:cochain_egenvalues}:
\be\label{eq:cocycle_egenvalues}
\hat{a}^q \ket{\eta^q} = \eta^q \ket{\eta^q}.
\ee
Finally, $\xi = d\zeta \in B^q$ is a $q$-coboundary, where $\zeta \in C^{q-1}$ is an arbitrary $(q-1)$-cochain.  Therefore, the $q$-cocycles  $\eta^q$ and $\eta^q+\xi$ are both in the same cocycle class $[\eta^q]$. According to Eq.~\eqref{eq:codeword_state}, the deformed classical codeword state can be written as
\be
\ket{\eta^q + \xi} = \big(\prod_{s_q \in \xi} X(s_q)\big) \ket{\eta^q},
\ee
where $\prod_{s_q \in \xi} X(s_q)$ is nothing but an $X$-stabilizer supported on $\xi$.
Note that one can also have the alternative definition for the logical-$Z$ eigensate in terms the action of the logical-$X$ operators on the all-zero logical-$Z$ eigenstate (with $n_\alpha=0$ for any $q$-cocycle $\alpha^q$):  
\be
\overline{\ket{\eta^{q}}} \equiv \Motimes_{\alpha}\overline{\ket{ n_\alpha  }}    =\lo{X}_{\eta^q} \Motimes_{\alpha}\lo{\ket{0}}_{\alpha}.
\ee

\begin{figure}[t]
\includegraphics[width=1\columnwidth]{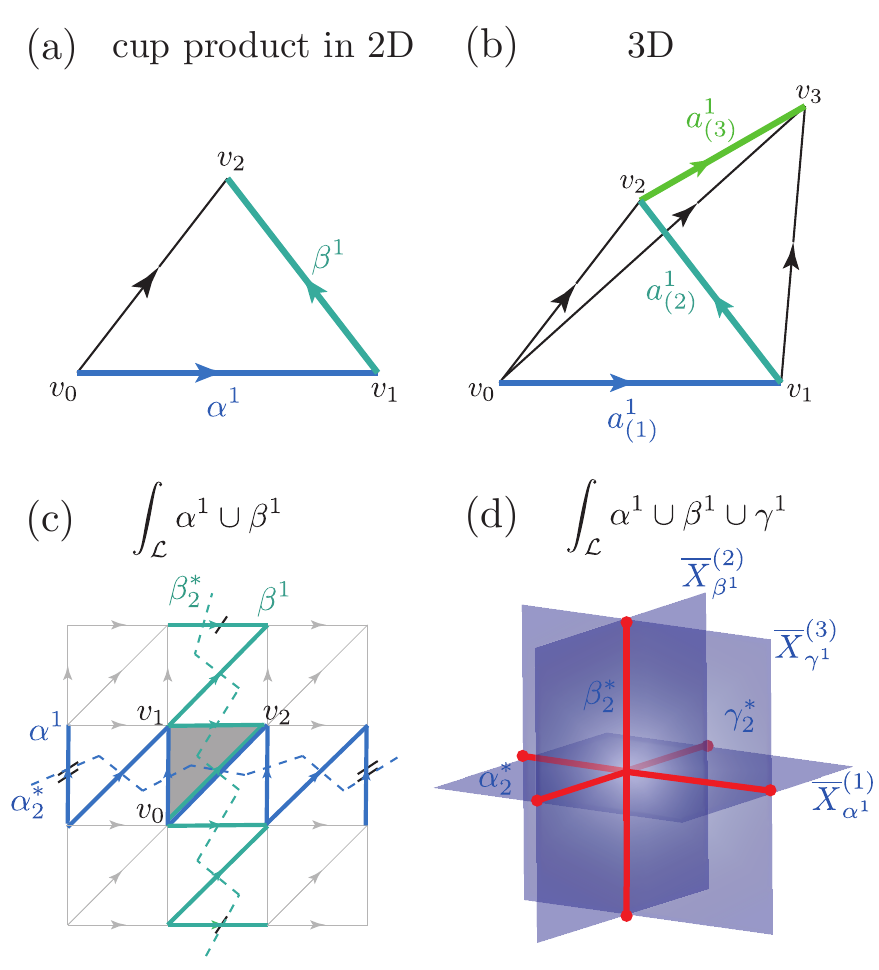}
\caption{(a)  Illustration of the rule of evaluating cup product between two 1-cochains on a 2-simplex, i.e., $(\alpha^1$$\cup$$ \beta^1)([v_0,v_1, v_{2}]) 
=\alpha^1([v_0,v_1])\beta^1([v_1, v_2])$. Note that the arrows on the edges point from vertices of lower order to those of higher order. (b) Illustration of evaluating a triple cup product of three operator-valued 1-cocyles on a 3-simplex, which gives a CCZ gates acting on three qubits located on edges $[v_0, v_1]$, $[v_1, v_2]$ and $[v_2, v_3]$ belonging to the three code copies $\C_{(1)}$, $\C_{(2)}$, and $\C_{(3)}$ respectively. (c) Illustration of the sum of cup product of three cocycles over the entire 2D triangulation $\L$, where the arrows indicate the vertex ordering. We can see that $\alpha^1 \cup \beta^1$ only evaluates to 1 on the highlighted 2-simplex $[v_0, v_1, v_2]$, while it evalutes to 0 for all the others 2-simplices.  The cup product sum evaluates the intersection number of the  Poincar\'e dual cycles (dashed) $\alpha^*_1$ and $\beta^*_1$.   (d) Illustration of the sum of triple cup product of three 1-cocyles $\alpha^1$, $\beta^1$ and $\gamma^1$ in 3D, which corresponds to the triple intersection of the three Poincare\'e dual 2-cycles $\alpha^*_2$, $\beta^*_2$ and $
\gamma^*_2$.  The 1-cocycles or their dual 2-cycles are the support of the logical-$X$ operators. }
\label{fig:cup_product}
\end{figure}

In order to study the logical gates, we then introduce the cup product `$\cup$' which corresponds to the following bilinear map on the cochain groups:
\be
\cup :  C^p(\L) \times C^q(\L) \rightarrow C^{p+q}(\L),
\ee 
where $C^p$ represents the $p^\text{th}$ cochain group.
This means the cup product between a $p$-cochain $\alpha^p \in C^{p}$ and a $q$-cochain $\beta^q \in C^{p}$ gives rise to a $(p+q)$-cochain  $\alpha^p \cup \beta^q \in C^{p+q}$.  

This cup product  can be explicitly evaluated on a $(p+q)$-simplex $[v_0,v_1,\cdots,v_{p+q}]$ as \cite{Hatcher:2001ut} 
\begin{align}\label{eq:cup_def}
 & (\alpha^p \cup \beta^q)([v_0,v_1,\cdots, v_{p+q}]) \cr
=&\alpha^p([v_0,v_1,\cdots, v_{p}])\beta^q([v_p, v_{p+1},\cdots, v_{p+q}])~. 
\end{align}
Here, we can choose an arbitrary ordering for the vertices $v_i$ on a ($p+q$)-simplex as $v_0 < v_1 < v_2 \cdots < v_{p+q}$, which hence specifies how to pick the $p$-simplices and $q$-simplicies in the above evaluation. The $p=q=1$ case is illustrated in Fig.~\ref{fig:cup_product}(a). Note that the cup product also induces a bilinear operation on the cohomology groups:
\be\label{eq:bilinear}
\cup :  H^p(\L) \times H^q(\L) \rightarrow H^{p+q}(\L).
\ee

We introduce a new unitary acting on  three copies of CSS codes $\C_{(1)}$, $\C_{(2)}$ and $\C_{(3)}$  (not necessarily identical copies), where the qubits are placed on the $q_1$-, $q_2$- and $q_3$-simplices respectively satisfying $q_1+q_2+q_3=n$. The associated electric gauge fields are operator-valued cochains $\hat{a}_{(1)}^{q_1}$, $\hat{a}_{(2)}^{q_2}$ and $\hat{a}_{(3)}^{q_3}$ respectively, which corresponds to the Pauli $Z$'s on the three copies of codes respectively.  The corresponding unitary can be expressed via the sum of triple cup product of the operator-valued cochains as 
\be\label{eq:cup_higher_form}
U = (-1)^{\int_{\L} \hat{a}_{(1)}^{q_1} \cup \hat{a}_{(2)}^{q_2} \cup \hat{a}_{(3)}^{q_3}}.
\ee
In this paper, we take the sum over all the $n$-simplices $s_n$ on the entire simplicial complex $\L$ which corresponds to the triangulation of a $n$-manifold $\M^n$ \footnote{More general case of simplicial complex beyond the manifold triangulation has been studied in Ref.~\cite{zhu2025topological}.}.
Note that the exponent is again a chain-cochain paring since it can be interpreted as the sum over an $n$-chain $\Sigma_n =\sum_{s_n} s_n$  consisting of all the $n$-simplices $s_n$, i.e., $\int_\L \bullet \equiv \int_{\Sigma_n} \bullet$ \ .  Note that since $\L$ tirangulates a closed manifold, $\Sigma_n$ is also an $n$-cycle, i.e., $ \partial  \Sigma_n =0$.   Note that the triple cup product induces a trilinear map on the homology groups:
\be\label{eq:trilinear_map}
\cup : H^{q_1}(\L) \times H^{q_2}(\L) \times H^{q_3}(\L) \rightarrow H^{q_1+q_2+q_3}(\L).
\ee

One can hence use the general rule in Eq.~\eqref{eq:cup_def} to evaluate the sum of cup product  as follows:
\begin{widetext}
\begin{align}\label{eq:logical_gate_higher_form}
\non U=& (-1)^{\int_{[v_0,v_1, \cdots, v_n] \in \L} \hat{a}_{(1)}^{q_1}([v_0,v_1, \cdots v_{q_1}])\hat{a}_{(2)}^{q_2}([v_{q_1}, \cdots, v_{q_1+q_2}]) \hat{a}_{(3)}^{q_3}([v_{q_1+q_2}, \cdots, v_n])} \\
=&\prod_{[v_0,v_1, \cdots, v_n] \in \L} \text{CCZ}^\text{(1,2,3)}([v_0,v_1, \cdots v_{q_1}],[v_{q_1}, \cdots, v_{q_1+q_2}], [v_{q_1+q_2}, \cdots, v_n]).
\end{align}
\end{widetext}
We can see that $U$ is a constant-depth circuit composed of a product of physical CCZ gates acting on a triple of qubits belonging to the three different copies of codes $\C_{(1)}$, $\C_{(2)}$ and $\C_{(3)}$ respectively.  For each $n$-simplex $[v_0,v_1, \cdots, v_n]$, the three qubits from three different code copies are located in the $q_1$-simplex $[v_0,v_1, \cdots, v_{q_1}]$, $q_2$-simplex $[v_{q_1}, \cdots, v_{q_1+q_2}]$ and $q_3$-simplex $[v_{q_1+q_2}, \cdots, v_n]$ respectively.  For the simplest case $q_1=q_2=q_3=1$ (3D simplicial complex), the CCZ gate for each 3-simplex $[v_0, v_1, v_2, v_3]$ is $\text{CCZ}^{(1,2,3)}([v_0, v_1],[v_1, v_2],[v_2, v_3])$, as illustrated in Fig.~\ref{fig:cup_product}(b).

We now present the following lemma:
\begin{lemma}\label{lemma_gate_2}
The unitary $U = (-1)^{\int_{\L} \hat{a}_{(1)}^{q_1} \cup \hat{a}_{(2)}^{q_2} \cup \hat{a}_{(3)}^{q_3}}$ acting on three copies of CSS codes defined on a $n$-simplicial complex $\L$ $(n=q_1+q_2+q_3)$ is a constant-depth local quantum circuit that implements collective logical CCZ gates.
\end{lemma}

\begin{proof}
We first check the action of $U$ on three copies of CSS codes defined on the triangulation $\L$, where the tensor product code space can be defined as
\begin{widetext}
\begin{align}\label{eq:CSS_state_triple}
\non \C_{(1)} \otimes \C_{(2)} \otimes \C_{(3)} :=& \text{Span} \bigg\{ \overline{\ket{\eta^{q_1}}}  \otimes \overline{\ket{\eta^{q_2}}}  \otimes \overline{\ket{\eta^{q_3}}} \propto  \sum_{\xi \in B^{q_1}} \sum_{\xi' \in B^{q_2}} \sum_{\xi'' \in B^{q_3}} \ket{\eta^{q_1} + \xi} \otimes \ket{\eta^{q_2} + \xi'} \otimes \ket{\eta^{q_3} + \xi''}\bigg| \ \eta^{q_1} \in H^{q_1}(\L; \Z_2), \\ 
& \eta^{q_2} \in H^{q_2}(\L; \Z_2), \eta^{q_3} \in H^{q_3}(\L; \Z_2) \bigg\}.
\end{align}
\end{widetext}

According to Eq.~\eqref{eq:cocycle_egenvalues},  each  classical codeword state in the above superposition is the eigenstate of the unitary $U$:
\begin{align}
& (-1)^{\int_{\L} \hat{a}_{(1)}^{q_1} \cup \hat{a}_{(2)}^{q_2} \cup \hat{a}_{(3)}^{q_3}}{\ket{\eta^{q_1}}}  \otimes {\ket{\eta^{q_2}}}  \otimes {\ket{\eta^{q_3}}} \cr
=&(-1)^{\int_{\L} \eta^{q_1} \cup \eta^{q_1} \cup \eta^{q_3}}  \ket{\eta^{q_1}  } \otimes \ket{\eta^{q_2} } \otimes \ket{\eta^{q_3} }, 
\end{align}
with the eigenvalue being a phase factor $(-1)^{\int_{\L} \eta^{q_1} \cup \eta^{q_1} \cup \eta^{q_3}}$.
Similarly,  we have
\begin{align}\label{eq:cup_shifted}
& (-1)^{\int_{\L} \hat{a}_{(1)}^{q_1} \cup \hat{a}_{(2)}^{q_2} \cup \hat{a}_{(3)}^{q_3}}{\ket{\eta^{q_1} + \xi}}  \otimes {\ket{\eta^{q_2}+ \xi'}}  \otimes {\ket{\eta^{q_3} + \xi''}} \cr
=&(-1)^{\int_{\L} (\eta^{q_1}+\xi ) \cup (\eta^{q_1}+\xi') \cup (\eta^{q_3}+\xi'')}  \cr
& \cdot {\ket{\eta^{q_1} + \xi}}  \otimes {\ket{\eta^{q_2}+ \xi'}}  \otimes {\ket{\eta^{q_3} + \xi''}},
\end{align}
for any coboundary $\xi \in B^{q_1}, \xi' \in B^{q_2}$ and $  \xi'' \in B^{q_3}$. According to the trilinear operation on cohomology  induced by the cup product  in Eq.~\eqref{eq:trilinear_map}, we know that
\be
(\eta^{q_1}+\xi ) \cup (\eta^{q_1}+\xi') \cup (\eta^{q_3}+\xi'') = \eta^{q_1} \cup \eta^{q_1} \cup \eta^{q_3}  +  d \omega,  
\ee
where $d\omega$ is a $n$-coboundary, i.e., $d^2 \omega =0$.  
The exponent in Eq.~\eqref{eq:cup_shifted} can hence be written as:
\begin{align}
& \int_{\L} (\eta^{q_1}+\xi ) \cup (\eta^{q_2}+\xi') \cup (\eta^{q_3}+\xi'')   \cr
=& \int_{\L} \eta^{q_1}  \cup \eta^{q_2} \cup \eta^{q_3}+ \int_{\L} d \omega^n \cr
=& \int_{\L} \eta^{q_1}  \cup \eta^{q_2} \cup \eta^{q_3}.
\end{align}
In the second equality, we have used the Stokes theorem that
\be
\int_{\L} d \omega^n = \int_{\partial \L}  \omega \equiv \int_{\partial \Sigma_n}  \omega^n = 0,
\ee
since $\Sigma_n$ is an $n$-cycle and hence has no boundary: $\partial \Sigma_n$$=$$0$.   Therefore, when acting $U$ on an arbitrary logical-$Z$ state $\overline{\ket{\eta^{q_1}}}  \otimes \overline{\ket{\eta^{q_2}}}  \otimes \overline{\ket{\eta^{q_3}}}$, the phase factor 
for each term in the superposition in Eq.~\eqref{eq:CSS_state_triple} is the same which then contributes to an overall phase factor, i.e., 
\be
U \overline{\ket{\eta^{q_1}}}  \otimes \overline{\ket{\eta^{q_2}}}  \otimes \overline{\ket{\eta^{q_3}}} = (-1)^{\int_{\L} \eta^{q_1} \cup \eta^{q_1} \cup \eta^{q_3}} \overline{\ket{\eta^{q_1}}}  \otimes \overline{\ket{\eta^{q_2}}}  \otimes \overline{\ket{\eta^{q_3}}}.
\ee
We can hence conclude that $U$ preserves the codes space $\C_{(1)} \otimes \C_{(2)} \otimes \C_{(3)}$.
Note that a more concise proof of this statement by the commutation relation with the stabilizer group (without using the specific form of the CSS code states) was also given in Ref.~\cite{zhu2023non, zhu2025topological, Hsin2024:classifying}.

Now we can re-express the circuit $U$ using the cohomology basis $\{\alpha^{q_1}\}$, $\{\beta^{q_2}\}$, and $\{\gamma^{q_3}\}$ to derive the corresponding logical gate:
\begin{align}\label{eq:logical_CCZ_higher}
 U =& \prod_{\alpha^{q_1}, \beta^{q_2}, \gamma^{q_3}} (-1)^{\int_{\L} (\hat{n}_\alpha  \alpha^{q_1} )\cup (\hat{m}_\beta  \beta^{q_2}) \cup (\hat{l}_\gamma  \gamma^{q_3})} \cr
   =&\prod_{\alpha^{q_1}, \beta^{q_2}, \gamma^{q_3}}\left[(-1)^{ \hat{n}_\alpha \hat{m}_\beta \hat{l}_\gamma}\right]^{\int_{\L} \alpha^{q_1} \cup \beta^{q_2} \cup \gamma^{q_3}} \cr
   =&\prod_{\alpha^{q_1}, \beta^{q_2}, \gamma^{q_3}} \overline{\text{CCZ}}[(\alpha^{q_1}; 1), (\beta^{q_2}; 2),(\gamma^{q_3}; 3)]^{\int_{\L} {\alpha^{q_1} \cup \beta^{q_2} \cup \gamma^{q_3}}}. \cr
\end{align}
The cup product sum in the exponent now corresponds to a 3-fold $\ZZ_2$ intersection number of the Poincar\'e dual cycles (denoted by `$*$'): 
\begin{align}\label{eq:intersection_poincare}
& \int_{\L} {\alpha^{q_1} \cup \beta^{q_2} \cup \gamma^{q_3}} =|\alpha^*_{n-q_1}  \cap \beta^*_{n-q_2} \cap  \gamma^*_{n-q_3}|.
\end{align}
Note that these cocycles and their Poincar\'e dual cycles are the support of the logical-$X$ operators, i.e., $\lo{X}^{(1)}_{\alpha^1}$, $\lo{X}^{(2)}_{\beta^1}$ and $\lo{X}^{(3)}_{\gamma^1}$ ($\lo{X}^{(1)}_{\alpha^*_2}$, $\lo{X}^{(2)}_{\beta^*_2}$ and $\lo{X}^{(3)}_{\gamma^*_2}$) in the original (dual) triangulations $\L$ ($\L^*$) respectively. The Poincar\'e duality corresponds to the isomorphim $H^i(\L; \ZZ_2)$$\cong$$ H_{n-i}(\L^*; \ZZ_2)$ where $n$ is the total space dimension.
\end{proof}
An illustration of the double intersection $\int_\L \alpha^1 \cup \beta^1$ and triple intersection $\int_\L \alpha^1 \cup \beta^1 \cup \gamma^1$ on a 2D and 3D triangulations are  illustrated in Fig.~\ref{fig:cup_product}(c) and (d) respectively.

For completeness and later use, we also introduce the conjugate cycle basis $\{\alpha_{q_1}\}$, $\{\beta_{q_2}\}$, and $\{\gamma_{q_3}\}$  of the above cocycle basis, which are the support of logical-$Z$ operators  $\{\lo{Z}_{\alpha_{q_1}}^{(1)}\}$, $\{\lo{Z}_{\beta_{q_2}}^{(2)}\}$ and $\{\lo{Z}_{\gamma_{q_3}}^{(3)}\}$. The correspondence between the conjugate pair of bases is due to the isomorphism between $\ZZ_2$ cohomology and homology $H^q(\L; \ZZ_2) \cong H_q(\L; \ZZ_2)$ given by the universal coefficient theorem \cite{Hatcher:2001ut}.  Note that in our notation $\alpha_{q_1}$ is the unique conjugate basis cycle of the basis cocycle $\alpha^{q_1}$ with the same label $\alpha$. The conjugate pair intersect on odd number of edges, as shown  by the following chain-cochain pairing:
\be
\int_{\alpha_q} \alpha^q \equiv \sum_{s_q} \alpha_q(s_q) \alpha^q(s_q) =1,
\ee
 The corresponding logical-$Z$ and -$X$ operators $\lo{Z}_{\alpha_q}$ and $\lo{X}_{\alpha^q}$ hence anticommute.  Note that a basis cycle and cocycle which are not in the same conjugate pair (i.e., with the same label) do not have non-trivial intersection, which can be expressed by the more general formula:
 \be\label{eq:conjugate_intersection}
\int_{\alpha'_q} \alpha^q \equiv \sum_{s_q} \alpha'_q(s_q) \alpha^q(s_q) =\delta_{\alpha, \alpha'}.
\ee

\subsection{Systolic geometry and freedom}\label{sec:systole}

We start with a set of definitions of systole in Riemannian geometries.

\begin{definition} \cite{Freedman_systole_2002}
    We define the Riemannian $q$-systole $(q \in \mathbb{N})$ of a Riemannian $r$-manifold $\M$ to be:
    \be
    sys_q(\M^r) = \inf_{\alpha_q \neq 0} \text{area}_q(\alpha_q),
    \ee
  where $\eta_q$ is a smooth oriented $q$-cycle belonging to a nontrivial $\Z$-homology class $[\eta_q] \neq 0 \in H_q(\M^r; \Z)$ and $\text{area}_q$ represents the $q$-area of $\eta_q$.  
\end{definition}

\begin{definition} \cite{Freedman_systole_2002}
    Similarly, we define the Riemannian $\Z_2$-$q$-systole $(q \in \mathbb{N})$ of a Riemannian $r$-manifold $\M^r$ to be:
    \be
    sys_q(\M^r; \ZZ_2) = \inf_{\eta_q \neq 0} \text{area}_q(\eta_q),
    \ee
  where $\eta_q$ is a smooth unoriented $q$-cycle belonging to a nontrivial $\Z_2$-homology class $[\eta_q] \neq 0 \in H_q(\M^r; \Z_2)$.  
\end{definition}

For a 2D Riemannian surface, we have 1-systoles defined to be the shortest essential loops (1-cycles). There is a famous systolic inequality proven by Loewner in 1949 for a torus $T^2$:
\be\label{eq:Loewner}
\frac{\left(sys_1(T^2)\right)^2}{area(T^2)} \le \frac{2}{\sqrt{3}},
\ee
where the equality holds only for the flat torus modled on a regular hexagon.
Moreover, the same inequality holds for the $\ZZ_2$ systole.  Note that this inequality essentially provides the upper bound on the distance scaling for the toric code, i.e., $d \le O(\sqrt{N})$, since the total number of qubits $N$ is proportional to the area of the torus and the distance $d$ is proportional to the 1-systole.

The l.h.s. of Eq.~\eqref{eq:Loewner} is called systolic ratio (SR).  More generally, for a Riemannian $r$-manifold $\M^r$ and $p+q=r$, we have the following definition:
\begin{definition}\cite{freedman:2020_manifold_from_code}
The $K$-$(p,q)$-Riemannian systolic ratio of $\M^r$ is defined to be:
\be
K\text{-}(p,q)\text{-}\text{SR}(\M^r) = \inf_{\eta_p, \eta_q \neq 0} \frac{\text{area}_p(\eta_p)\cdot \text{area}_q(\eta_q)}{vol(\M^r)},
\ee
where $\eta_p$ and $\eta_q$ belong to the non-trivial homology classes $[\eta_p] \neq0 \in H_p(\M^r; K)$ and $[\eta_q] \neq 0 \in H_q(\M^r; K)$, and $K$ is a ring. 
\end{definition}
\nin Note that for our purpose, we consider $K=\Z$ (integer) and $K=\ZZ_2$ (integer mod 2).

Since our present paper focuses on quantum codes, we will also need to consider the combinatorial systolic geometry defined on the triangulation (cellulation) of a manifold.  We now define the combinatorial $\ZZ_2$ systole in the following, which is what we focus on in this paper.  For that, one can essentially replace all the above definitions by the following replacement:
\be
area_q(\eta_q) = |\eta_q|, \quad vol(\M^r) = |\L|_r
\ee
Here $|\eta_q|$ counts the number of $q$-simplices $s_q$ supported on $\eta_q$ with non-zero coefficient, i.e., the Hamming weight, as defined in Eq.~\eqref{eq:Hamming_weight}. The combinatorial version of the volume $vol(\M^r)$ is counting the number of (top) $r$-simplices in the triangulation $\L$ of the manifold $\M^r$,  denoted by $|\L|_r$.  More details about the correspondence between the Riemannian and combinatorial geometries can be found in Refs.~\cite{Freedman_systole_2002, freedman:2020_manifold_from_code}.
\begin{definition} \cite{Freedman_systole_2002}
    We define the combinatorial $\Z_2$-$q$-systole $(q \in \mathbb{N})$ of a triangulated $r$-manifold $\M^r$ to be:
    \be
    sys_q(\M^r; \ZZ_2) = \inf_{\eta_q \neq 0} |\eta_q|,
    \ee
  where the q-cycle $\eta_q$ belongs to a nontrivial $\Z_2$-homology class $[\eta_q] \neq 0 \in H_q(\M^r; \Z_2)$. 
\end{definition}
\nin We can have a similar definition for the systole with $\Z$ coefficients.

For later use, we also define the cosystole as:
\begin{definition} 
    We define the combinatorial $\Z_2$-$q$-cosystole  of a triangulated $r$-manifold $\M^r$ to be:
    \be
    sys^q(\M^r; \ZZ_2) = sys_{r-q}(\M^r; \ZZ_2)= \inf_{\eta^q \neq 0} |\eta^q|,
    \ee
  where the q-cocycle $\eta^q$ belongs to a nontrivial $\Z_2$-cohomology class $[\eta^q] \neq 0 \in H^q(\M^r; \Z_2)$. 
\end{definition}
\nin Note that the first equality in the above equation is due to Poincar\'e duality. Here, $|\eta^q|=\sum_{s_q}|\eta^q(s_q)|$ is the Hamming weight of the cocycle  [see Eq.~\ref{eq:cocycle_Hamming_weight}]. 

Compared to Eqs.~\eqref{eq:Hamming_weight} and \eqref{eq:cocycle_Hamming_weight}, we can see the correspondence of the $q$-sysole and $q$-cosystole for a $(q,r-q)$-homological code:
\be
d_Z= sys_q(\M^r; \ZZ_2); \quad  d_X=sys^q(\M^r; \ZZ_2).
\ee

Now the combinatorial version of the systolic ratio is defined as follows:
\begin{definition}\cite{freedman:2020_manifold_from_code}
The $K$-$(p,q)$-combinatorial systolic ratio of a manifold $\M^r$ $(p+q=r)$ with triangulation $\L$ is defined to be:
\be\label{eq:combinatorial_SR}
K\text{-}(p,q)\text{-}\text{SR}(\M^r) = \inf_{\eta_p, \eta_q \neq 0} \frac{|\eta_p|\cdot |\eta_q|}{|\L|_r},
\ee
where $\eta_p$ and $\eta_q$ belong to the non-trivial homology classes $[\eta_p] \neq0 \in H_p(\M^r; K)$ and $[\eta_q] \neq 0 \in H_q(\M^r; K)$, and $K$ is a ring. 
\end{definition}
Now for both the Riemannian and combinatorial cases, we can define the systolic freedom as follows: 
\begin{definition}\cite{freedman:2020_manifold_from_code}
An $r$-manifold $\M^r$ has a $K$-$(p, q)$-systolic freedom, $p+q=r$, if $\M^r$ admits a sequence $\{i\}$ of Riemannian metrics (triangulations) so that 
\be
K\text{-}(p,q)\text{-}\text{SR}(\M^r_{(i)}) \rightarrow  \infty.
\ee

\end{definition}
One can further quantify the systolic freedom as follows:
\begin{definition}\cite{freedman:2020_manifold_from_code}\label{def:power-law}
An $r$-manifold $\M^r$ has a power-law $K$-$(p, q)$-systolic freedom if 
\be
K\text{-}(p,q)\text{-}\text{SR}(\M^r_{(i)}) = \Omega(vol(\M^r_{(i)})^\alpha),
\ee
for some $\alpha >0$.

\end{definition}
For the oriented case, there are many examples for a power-law $\ZZ$-systolic freedom, see e.g.,~Ref.~\cite{Gromov:1996}.  However, for many of these examples, the $\ZZ_2$-systolic freedom does not exist since the unoriented cycles can be much shorter. Therefore, the $\ZZ_2$-systolic freedom is usually harder to achieve, and it was even conjectured by Gromov that $\ZZ_2$-systolic freedom did not exist \cite{Katz:1998}.  The construction in Ref.~\cite{Freedman_systole_2002} gave the first counter-example to this conjecture, but only exhibits a weak  polylog $\ZZ_2$-systolic freedom: $\ZZ_2$-$(1,2)$-$SR=\Omega\left(log^{1/2}(vol(\M^3))\right)$.  The power-law systolic freedom was only discovered two decades later in the recent work by Freedman and Hastings \cite{freedman:2020_manifold_from_code} using their code-to-manifold mapping with the input of the recent qLDPC codes breaking the $\sqrt{N}\log^{1/2}(N)$ distance barrier \cite{fiberbundlecode21, PK:almost_linear, Breuckmann:2021_balanced,  pkldpc22}.

For the purpose of transversal non-Clifford gate, we will construct a manifold not only having power-law systolic freedom, but also the with the presence of non-trivial triple intersection between cycles with large systoles [see Eq.~\eqref{eq:intersection_poincare}], which was not known to exist.   For example, the manifold constructed in Ref.~\cite{freedman:2020_manifold_from_code} is coarsely 2D according to Ref.~\cite{portnoy2023local} \footnote{This means the manifold can be embedded in a 2D non-Euclidean space.} and hence cannot support a non-trivial triple intersection between large cycles.

\section{Construction of the triple good subsystem codes}
\label{sec:subsystem}

In this section we introduce the triple good subsystem codes formed by a triple product of manfolds built from good qLDPC codes, which admits transversal CCZ gates.  

The key of the code-to-manifold mapping in  Ref.~\cite{freedman:2020_manifold_from_code} is to first lift the $\ZZ_2$-chain complex of the input quantum code to a $\Z$-chain complex.  One can then turn the qubits and checks in the input quantum code into dressed handles, and attach the handles with different indices according to the lifted boundary map. In the lowest dimensional example (11D), one has the following cellular chain complex to describe the handle attachment:
\begin{align}
&C_{11} \rightarrow \cdots \rightarrow C_6 \rightarrow  C_{5} \xrightarrow[]{\partial_{5}=\mathbf{H}_Z^T} C_4 \xrightarrow[]{\partial_{4}=\mathbf{H}_X} C_{3} \rightarrow C_2 \rightarrow \cdots, \cr
&\qquad   \quad \quad \quad \qquad  Z\text{-check} \qquad \ \text{qubit} \quad \ \   X\text{-check}    \cr
\end{align}
where the $X$-checks, qubits and $Z$-checks correspond to the dressed 3-handles, 4-handles and 5-handles respectively. 
The attachement of these handles gives rise to a handle-body $H$, and one can then take the double of the handlebody, i.e., glue the handlebody with an identical copy along their common boundary with an identity map:  $\M=\D H =H \cup_{id_{\partial H}} H$.  This completes the construction of a closed manifold $\M$.  The details of this construction can be found in Refs.~\cite{freedman:2020_manifold_from_code, zhu2025topological}.  There is also an excellent pedagogical review of this code-to-manifold mapping in Ref.~\cite{guemard2025lifting}.

Let us first introduce the following theorem essentially obtained from Ref.~\cite{freedman:2020_manifold_from_code} (Theorem 1.2.1) with the additional input from Ref.~\cite{pkldpc22},  which was also used to construct the 3D local code in Ref.~\cite{portnoy2023local} (restated as Theorem~5) \footnote{\label{footnote:freedman}As has been clarified below Theorem 5 in Ref.~\cite{portnoy2023local}, the original Theorem 1.2.1 in Ref.~\cite{freedman:2020_manifold_from_code} has a $polylog(m)$ reduction in the rate and distance due to the additional requirement that the underlying manifold is simply connected for the interest of systolic geometry.  When dropping this additional requirement which is unnecessary for the present paper, the proof in Ref.~\cite{freedman:2020_manifold_from_code} gives the optimal parameters without the $polylog(m)$ reduction.}: 
\begin{theorem}\label{theorem:FH} $($Freedman and Hastings \cite{freedman:2020_manifold_from_code}$)$ \textbf{Building manifolds from quantum codes}: 
Given the good qLDPC code $\bar{\C}$ from Ref.~\cite{pkldpc22} with the parameters $[[n, \Theta(n), \Theta(n)]]$ as an input and for any chosen dimensions $q \ge 4$ and $r \ge 2q +3 $, there exists a mapping from the input code $\bar{\C}$ to an associated triangulated $r$-dimensional $(r\ge 11)$ manifold $\M^{r}$ satisfying the following properties:
\begin{enumerate}
\item 
$\M^r$ has a bounded local geometry, i.e., each vertex in its triangulation $\L$ is adjacent to $O(1)$ simplices;  
\item 
$vol(\M^{r})= \Theta(n)$;
\item
$b_q=\text{dim}(H_q(\M^r; \ZZ_2)) = b_{r-q}= \Theta(n)$;
\item 
$sys_q(\M^r; \ZZ_2)=sys^{r-q}(\M^r; \ZZ_2) = \Theta(n)$,  \\
$sys^q(\M^r; \ZZ_2)=sys_{r-q}(\M^r; \ZZ_2) = \Theta(n)$.
\end{enumerate}
\end{theorem}
Note that the above theorem also gives rise to the existence of a thickened $(q, r-q)$-homological qLDPC code defined on the triangulation $\L$ of the manifold $\M^r$ with code parameters $[[\Theta(n), \Theta(n), \Theta(n)]]$.
The lowest-dimensional example is $q=4$ and $r=11$, corresponding to a thickened $(4, 7)$-homological qLDPC code defined on the triangulation  of the 11-manifold $\M^{11}$.

We now construct a product manifold $\widetilde{\M}=\M^{11} \times \M'^{11} \times \M''^{11} $ with dimension $33$, along with the corresponding product simplicial complex $\widetilde{\L}$.
We construct three copies of qLDPC codes defined on $\widetilde{\M}$ with qubits placed on $11$-simplices.  Let the Kunneth map be 
\begin{align}
K: & H^*(\M^{11}; \Z_2) \otimes H^*(\M'^{11}; \Z_2) \otimes H^*(\M''^{11}; \Z_2) \cr
\non &\rightarrow H^*(\widetilde{\M}, \Z_2).
\end{align}
We then use the above map to express the following cohomology classes in terms of the tensor product of cocyles in each constituent manifold ($\M^{11}$, $\M'^{11}$ and $\M''^{11}$):
\begin{align}\label{eq:cocycle_class1}
\alpha^{11} =& K(\as^4 \otimes {{\as^{*}}'}^7 \otimes \cs''^0 ) \equiv  \as^4 \otimes {{\as^{*}}'}^7 \otimes \cs''^0    \cr
\beta^{11} =& K(\cs^0  \otimes \as'^4  \otimes  {{\as^{*}}''}^{7}) \equiv
\cs^0  \otimes \as'^4  \otimes  {{\as^{*}}''}^{7}
\cr
\gamma^{11} =& K({\as^{*}}^7  \otimes \cs'^0 \otimes {{\as}''}^{4}) \equiv {\as^{*}}^7  \otimes \cs'^0 \otimes {{\as}''}^{4},  
\end{align}
where we have omitted $K$ for conciseness.  
These cocycles will be the support of the logical-$X$ operators.  

The logical-$Z$ operators are supported on the conjugate cycles:
\begin{align}\label{eq:cycle_class1}
\alpha_{11} =& \as_4 \otimes {{\as_7^{*}}'} \otimes \cs''_0    \cr
\beta_{11} =& \cs_0 \otimes \as'_4  \otimes {{\as_7^{*}}''}  \cr
\gamma_{11} =& \as_7^{*} \otimes \cs_0' \otimes \as_4''. 
\end{align}
Note that in our notation cycles and cocycles with the same label, such as $(\alpha_{11}, \alpha^{11})$, $(\as_4, \as^4)$ are conjugate pairs, which satisfy the intersection condition in Eq.~\eqref{eq:conjugate_intersection}. In addition, $\as^*_7$ is the Poincar\'e dual cycle of the cocycle $\as^4$ with complementary dimension (7+4=11), which will be denoted by $\as^4 \sim \as^*_7$. Their support on the manifold is essentially the same, with one in the original triangulation $\L$ and the other in the dual triangulation $\L^*$. Similarly, ${\as^{*}}^7$ is the Poincar\'e dual cocycle of the cycle $\as_4$.   We also say ${\as^{*}}^7$ and $\as^4$ are a pair of dual cocycles with complementary dimensions, which has the following cup product  (intersection) property:
\be
\int_{\M} \as^4 \cup {\as^{*}}^7 = |\as_7^{*} \cap \as_4|=1.
\ee
Similarly, we say $\as_7^{*}$ and $\as_4$ are a pair of dual cycles.
 
 Note that $\alpha^{11}$, $\beta^{11}$, and $\gamma^{11}$ and their conjugate cycles $\alpha_{11}$, $\beta_{11}$, and $\gamma_{11}$ are not the only contributions to the cocycles and cycles of dimension 11 when using the Kunneth formula.   

For example, the following types of cocycle classes and their conjugate cycle classes will have the dimension equaling 11:
\begin{align}
& {\cs^{*}}^{11} \otimes \cs'^0 \otimes \cs''^0, \quad  \cs^0 \otimes {{\cs^*}'}^{11} \otimes \cs''^0, \quad \cs^0  \otimes \cs'^0 \otimes {{\cs^*}''}^{11};  \cr
\non & \cs^{*}_{11} \otimes \cs'_0 \otimes \cs''_0, \quad  \cs_0 \otimes {\cs_{11}^{*'}} \otimes \cs''_0, \quad \cs_0  \otimes \cs'_0 \otimes {\cs_{11}^{*''}}.
\end{align}
Here,  $\cs^0$ is the emergent 0-cocycle which is the Poincar\'e dual of the $11$-cycle $\cs^*_{11}$ wrapping around the entire manifold $\M^{11}$, while ${\cs^{*}}^{11}$ is the dual $11$-cocycle of the 0-cycle $\cs_0$.  Similar relations hold for the second and third manifolds $\M'^{11}$ and $\M''^{11}$.  The main issue of the above three cycle classes is that the corresponding systoles $sys_{11}(\widetilde{\M}; \ZZ_2)$ have size $O(n)=O(N^{\frac{1}{3}})$ due to the presence of two 0-cycles with $O(1)$ size, e.g., $\cs'_0$ and $\cs''_0$, in the product, which is smaller than the desired code distance $\Omega(N^{\frac{2}{3}})$.   

In order to resolve this issue, we can use a subsystem-code encoding which selects only  a subset of cycles and their conjugate co-cycles to encode logical $Z$ and $X$ operators, while treating the rest of logical degree of freedom associated with logical operators of shorter cycles as \textit{gauge qubits}.

To gain further insight, consider the example shown in Fig.~\ref{fig:sub_system_encoding}.  We start with a torus $T^2$ having 1-systole of size $L$, i.e., $\text{min}(|\alpha_1|)=\text{min}(|\beta_1|)=L$, where $\alpha_1$ and $\beta_1$ represent the longitudinal and meridian 1-cycles. Modify this surface by excising two small disks and attaching a narrow handle between them, yielding a genus-2 surface $\Sigma_2$. This modification introduces an additional pair of 1-cycles, $\alpha'_1$ and $\beta'_1$, each of constant length $O(1)$. A natural homology basis for $\Sigma_2$ is thus $\{\alpha_1, \beta_1, \alpha'_1, \beta'_1 \}$.

This setup supports four logical qubits, with each qubit defined by a pair of dual logical operators supported on intersecting basis cycles—e.g., $\lo{X}_{\alpha_1}$ and $\lo{Z}_{\beta_1}$, where $|\alpha_1 \cap \beta_1| = 1$, as depicted in Fig.~\ref{fig:sub_system_encoding}. Due to the existence of the short cycles $\alpha'_1$ and $\beta'_1$, the minimal nontrivial cycle on the surface now has length $O(1)$, so the 1-systole of $\Sigma_2$ is also $O(1)$. If one were to define a conventional subspace code over the full homology group $H_1(\Sigma_2; \mathbb{Z}_2)$, the resulting code would have four logical qubits and a code distance $d = O(1)$.

However, by interpreting the logical qubits associated with $\alpha'_1$ and $\beta'_1$ as \textit{gauge qubits}—which do not store protected quantum information—we can restrict our attention to a subsystem code supported only on the long cycles $\alpha_1$ and $\beta_1$. The subsystem code has two logical qubits with operator pairs $(\lo{X}_{\alpha_1}, \lo{Z}_{\beta_1})$ and $(\lo{Z}_{\alpha_1}, \lo{X}_{\beta_1})$. Importantly, any Pauli error confined to the short cycles does not intersect the operators on the long cycles and therefore cannot induce logical errors in this subsystem encoding. As a result, the effective code distance of the subsystem code remains $d = L$, determined by the lengths of the large cycles.

\begin{figure}[hbt]
\includegraphics[width=1\columnwidth]{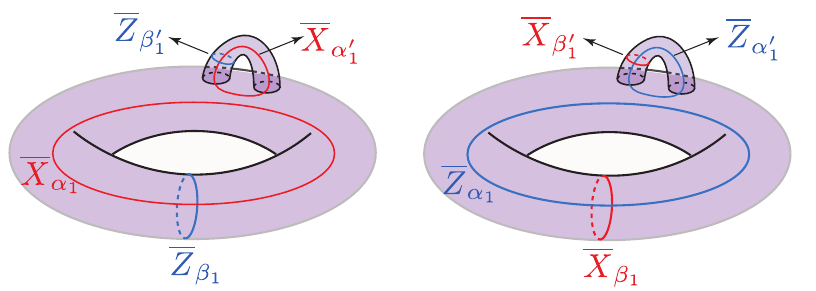}
\caption{An illustration of the subsystem encoding for a homological code defined on a manifold. Given a chosen homology basis, two logical qubits are encoded using a dual pair of intersecting cycles, $\alpha_1$ and $\beta_1$, each of length $O(L)$. Another two logical qubits are encoded using a second dual pair, $\alpha'_1$ and $\beta'_1$, whose cycle lengths are only $O(1)$. By designating the qubits associated with the short cycles as gauge qubits, quantum information is not stored in them. As a result, the remaining two logical qubits retain a large subsystem code distance of $O(L)$.}\label{fig:sub_system_encoding}
\end{figure}

For general situations, we have the following lemma:
\begin{lemma}\cite{zhu2025topological}\label{lemma:subsystem}
    For a homological quantum code defined on the triangulation of a  $k$-manifold $\M^k$, one can define a subsystem code by associating the logical-$Z$ operators with a subset of an $i^\text{th}$ homology basis $\{\alpha_i\}$ and the conjugate logical-$X$  operators on the dual subset of $(k-i)^\text{th}$ homology basis $\{\beta^*_{k-i}\}$ satisfying the intersection relation $|\alpha_i \cap \beta^*_{k-i}|= \delta_{\alpha, \beta}$.  The distance of the subsystem code is hence $d=\text{min}(\text{min}\{|\alpha_i|\}, \text{min}\{|\beta^*_{k-i}|\})$.
\end{lemma}
See Appendix \ref{app:proof} for the proof of this lemma. \\

We now present the following theorem: 
\begin{theorem}\label{theorem:subsystem}
There exist a family of triple good subsystem code $\tilde{\C}$ defined on the triangulation of a $33$-dimensional product manifold  $\widetilde{\M}=\M^{11} \times \M'^{11} \times \M''^{11}$  with dimension $K=\Theta(N^{\frac{2}{3}})$,  code distance $d=\Omega(N^{\frac{2}{3}})$ and constant stabilizer weight $w=O(1)$,  such that a constant-depth circuit implementing the  triple cup product gives rise to $\Theta(N)$ logical CCZ gates on three copies of $\tilde{\C}$.     
\end{theorem}

\begin{proof}
As stated above, the three copies of quantum codes supported on $\widetilde{M}$  have logical-$X$ and $Z$ operators associated with cocycles and cycles with dimension 11 respectively.    

The dimension (number of logical gates) of the three identical code copies, which can be obtained by counting the sum of the dimensions of the set of 11-cocycle classes $\{\alpha^{11}\}$, $\{\beta^{11}\}$, and $\{\gamma^{11}\}$.  The dimension of the $\{\alpha^{11}\}$  can be decomposed as the product of the dimension of the set of cocycle classes in each tensor component:
\begin{align}
 \dim (\{\alpha^{11}\})=& \dim (\{\as^4\}) \cdot \dim (\{{{\as^{*}}'}^7 \}) \cdot \dim (\{{\cs''}^0\})
\cr
=& \Theta(n)  \cdot \Theta(n) \cdot 1 = \Theta(n^2) = \Theta(N^{2/3}), \cr 
\cr
\end{align}
where we have used the fact that there is only a unique 0-cocycle $\cs''^0 \sim {\cs_{11}^*}''$ which is supported on the entire manifold $\M'^{11}$.   Similarly, one also obtains
\begin{align}
 \dim (\{\beta^{11}\})=& \dim (\{{\cs^0}\}) \cdot \dim (\{{\as'}^4\}) \cdot \dim (\{{{\as^{*}}''}^{7}\})
\cr
=&  1 \cdot \Theta(n) \cdot \Theta(n)  = \Theta(n^2) = \Theta(N^{2/3}), \cr 
\dim (\{\gamma^{11}\})=& \dim ( \{{\as^{*}}^7 \}) \cdot \dim( \{ \cs'^0\}) \cdot \dim (\{{\as''}^{4}\})
\cr
=& \Theta(n) \cdot 1   \cdot \Theta(n)  = \Theta(n^2) = \Theta(N^{2/3}).
\end{align}
The total code dimension is hence
\be
K= \dim (\{\alpha^{11}\}) + \dim (\{\beta^{11}\}) + \dim (\{\gamma^{11}\})=\Theta(N^{2/3}).
\ee

We have the following lower bound on the minimal size of the subset of 11-cocycle classes $\{|\alpha^{11}|\}$:
\begin{align}
 \min \{|\alpha^{11}|\} =& \min\{|\as^4 \otimes {{\as^{*}}'}^7 \otimes \cs''^0 |\}  \cr
\ge & \min \{|\as^{4}|\} \cdot \min\{|{{\as^{*}}'}^7|\} \cdot \min \{|\cs''^0|\}  \cr
=& \Omega(n) \cdot \Omega(n) \cdot  \Omega(n) = \Omega(n^3) = \Omega(N), \cr
\end{align}
where we have used the fact that $|\cs''^0|=|{\cs_{11}^*}''|=\Omega(n)$ since $\cs''^0$ supports on the entire 11-manifold $\M'^{11}$ and the fact that the logical-$X$ and -$Z$ distances of the constituent good qLDPC codes corresponding to $\min\{|\as^{4}|\}$ and $ \min\{|{{\as^{*}}'}^7|\}$ respectively are $\Omega(n)$.   Similarly, we can get the following lower bound on the other two subsets of 11-cocycle classes 
\be
\text{min}\{|\beta^{11}|\} =\Omega(N),  \quad \text{min}\{|\gamma^{11}|\} =\Omega(N).
\ee
Therefore, the $X$-distance of the code is linear, i.e., 
\be
d_X=\min( \min \{|\alpha^{11}|\}, \text{min} \{|\beta^{11}|\}, \text{min}\{|\gamma^{11}|\}) = \Omega(N).
\ee

For the subset of 11-cycle classes $\{\alpha_{11}\}$, we have 
\begin{align}
\min \{|\alpha_{11}|\} =& \min \{|\as_4 \otimes {{\as_7^{*}}'} \otimes \cs''_0  |\}  \cr
\ge & \min\{|\as_{4}|\} \cdot \min \{|{\as_7^{*}}'|\} \cdot \min\{|\cs''_0|\}  \cr
=& \Omega(n) \cdot   \Omega(n) \cdot \Omega(1)  = \Omega(n^2) = \Omega(N^{\frac{2}{3}}), \cr
\end{align}
where we have used the fact that $|\cs''_0|=1$ since the minimal-size representative of $\cs''_0$ is just a single vertex,  and the fact that the logical-$Z$ and -$X$ distances of the constituent good qLDPC codes corresponding to $\min\{|\as_{4}|\}$ and $ \min\{|{{\as_7^{*}}'}|\}$ respectively are $\Omega(n)$. 
Similarly we have
\be
 \text{min}\{|\beta_{11}|\} =\Omega(N^{\frac{2}{3}}), \quad  \text{min}\{|\gamma_{11}|\} =\Omega(N^{\frac{2}{3}}).
\ee
Therefore, the $Z$-distance of the code is 
\be
d_Z=\min( \min \{|\alpha_{11}|\}, \text{min} \{|\beta_{11}|\}, \text{min}\{|\gamma_{11}|\}) = \Omega(N^{\frac{2}{3}}).
\ee
We hence obtain the code distance as
\be
d = \min (d_X, d_Z)= \Omega(N^{\frac{2}{3}}).  
\ee


We now apply the constant-depth circuit $U$ to obtain the collective logical CCZ.   According to Eq.~\eqref{eq:logical_CCZ_higher},  we can apply the following unitary 
\begin{align}
U =& (-1)^{\int_{\L} \hat{a}_{(1)}^{11} \cup \hat{a}_{(2)}^{11} \cup \hat{a}_{(3)}^{11}}  \cr
=& \prod_{\alpha^{11}, \beta^{11}, \gamma^{11}} \overline{\text{CCZ}}[(\alpha^{11}; 1), (\beta^{11}; 2),(\gamma^{11}; 3)]^{\int_{\L} {\alpha^{11} \cup \beta^{11} \cup \gamma^{11}}}, \cr
\end{align}
where $a^{11}_{(i)}$ is operator-valued 11-cochains (gauge fields).
Note that non-trivial logical CCZ will be applied to the logical qubits with the labels $(\alpha^{11}; 1)$, $(\beta^{11}; 2)$ and $(\gamma^{11}; 3)$ as long as the triple cup product sum $\int_{\L} {\alpha^{11} \cup \beta^{11} \cup \gamma^{11}}$ evaluates to $1$.

We can now decompose the triple cup product via the Künneth formula as follows:
\begin{align}\label{eq:triple_decomposition}
&  {\alpha^{11} \cup \beta^{11} \cup \gamma^{11}} \cr
=&   (\as^4 \otimes {{\as^{*}}'}^7 \otimes \cs''^0 ) \cup (\cs^0  \otimes \as'^4  \otimes  {{\as^{*}}''}^{7}) \cup ( {\as^{*}}^7  \otimes \cs'^0 \otimes {{\as}''}^{4}) \cr
=&  ( \as^4 \cup  \cs^0 \cup \as^{*7}) \otimes  ( {{\as^*}'}^7 \cup  \as'^4 \cup \cs'^0 ) \otimes  ( \cs''^0  \cup {{\as^{*}}''}^7 \cup {\as''}^4 )  \cr
\neq& 0.
\end{align}
The tensor component of  $ \as^4 \cup  \cs^0 \cup \as^{*7} = \as^4 \cup   \as^{*7} \cup \cs^0 $ is non-trivial due to the fact that $\as^4 \cup \as^{*7}$ is non-trivial according to the Poincar\'e duality and the fact that an additional intersection with the $0$-cocycle $\cs^0 \sim \cs_{11}$ (supported on the entire 11-manifold) is still non-trivial.  Similarly, all the three tensor components in  Eq.~\eqref{eq:triple_decomposition} is non-trivial, which gives rise to the non-trivial  triple cup product $ {\alpha^{11} \cup \beta^{11} \cup \gamma^{11}}$. 
One can then evaluate the cup prodcut sum over the 33-manifold as 
\begin{align}\label{eq:simplicial_Kunneth}
& \int_{\widetilde{\M}^{33}} {\alpha^{11} \cup \beta^{11} \cup \gamma^{11}} \cr
=&  \int_{\M^{11}} ( \as^4 \cup  \cs^0 \cup \as^{*7} ) \cdot \int_{\M'^{11}} ({{\as^*}'}^7 \cup  \as'^4 \cup \cs'^0 ) \cr
& \cdot \int_{\M''^{11}} (\cs''^0  \cup {{\as^{*}}''}^7 \cup {\as''}^4 )\cr
=& 1.
\end{align}
In the above derivation, we have used the non-trivial triple intersection within each constituent 11-manifold such as 
\be
\int_{\M^{11}}  \as^4 \cup  \cs^0 \cup \as^{*7}  = \int_{\M^{11}} \as^4 \cup \as^{*7} \cup \cs^0 = |\as^*_7 \cap \as_4 \cap \cs^*_{11}| = 1, 
\ee
which comes from the intersection between the cocycles and their Poincar\'e dual cycles  $\as^4 \sim \as^*_7$ and $\as^{*7} \sim \as_4$ at a single point, which in turn intersects with the entire constituent  11-manifold $\cs^0 \sim \cs_{11}$ at a single point.  

Since one can have $\Theta(n)$ choices for $\as^4$, $\as'^4$ and $\as''^4$ respectively to get non-trivial triple cup product, there are hence in total $\Theta(n^3)=\Theta(N)$ triple intersection points.   This gives rise to $\Theta(N)$ logical CCZ gates

\end{proof}

\section{Construction of the triple good  (subspace) codes and the corresponding exotic manifolds with systolic freedom}\label{sec:subspace}

In the last section,  we have used the simple construction of the subsystem qLDPC code to get large subsystem distance $\Omega(N^{\frac{2}{3}})$ despite the presence of smaller 11-systole of size $O(N^{\frac{1}{3}})$  due to the emergent cycle such as $\cs^*_{11} \otimes \cs'_0 \otimes \cs''_0$.  For practical purpose, this is good enough, but for conceptual understanding it would be nice to figure out whether one can still obtain $\Omega(N^{\frac{2}{3}})$ for a CSS subspace code supporting non-Clifford gates.  Geometrically, this is equivalent to requiring one to construct a manifold with $\ZZ_2$ systolic freedom which also contains $\ZZ_2$ triple intersection points.   

In the following, we will construct a product manifold $\widetilde{\M}=\M^{r_1} \times \M'^{r_2} \times \M''^{r_3} $ with total dimension $r_1+r_2+r_3=3q$, along with the corresponding product simplicial complex $\widetilde{\L} =\L^{r_1} \otimes \L'^{r_2} \otimes \L''^{r_3} $.
We then construct the subspace qLDPC code defined on $\widetilde{\M}$ with qubits placed on $q$-simplices and call them the `\textit{triple good code}'.  

First, we can obtain the following theorem 
\begin{theorem}\label{theorem:3q-manifold} 
For some dimension $q \ge 31$,  there is a triangulated $3q$-dimensional manifold $\widetilde{\M}$ with $N$ vertices and a power-law $\ZZ_2$-$(q, 2q)$-systolic freedom which has the following properties:

\begin{enumerate}
\item The triangulation has bounded geometry i.e., each vertex is adjacent to $O(1)$ simplices.
\item $vol(\widetilde{M})= \Theta(N)$.
\item  $\tilde{b}_q=dim(H_q(\widetilde{\M}; \Z_2)) = \Theta(N^{\frac{2}{3}})$.
\item $sys_q(\widetilde{\M}; \Z_2) = sys^{2q}(\widetilde{\M}; \Z_2) =\Omega(N^{\frac{2}{3}})$, \\ $sys^{q}(\widetilde{\M}; \Z_2) = sys_{2q}(\widetilde{\M}; \Z_2) =  \Omega(N^{\frac{2}{3}})$.
\item There exist $\Theta(N)$ triple intersection points for any chosen basis of $2q$-cycles. 
\end{enumerate}
\end{theorem}

\begin{proof}
According to Theorem \ref{theorem:FH}, for any two integers $p$ and $s$ with  $p \ge 4$ and  $s \ge p +3$,  we can construct a triangulated $r$-dimensional manifold $FH(p,s)$ of bounded geometry with $n$ vertices \cite{freedman:2020_manifold_from_code}, where $r=p+s$. Recall that this manifold is constructed by modeling a good quantum LDPC code \cite{pkldpc22} with a handlebody construction.   This manifold $\M=FH(p,s)$ has $p$ and $s$ dimensional systole and cosystoles of size
\begin{align}
sys_p (\M; \Z_2)=sys^s (\M; \Z_2)= \Theta(n), \cr
sys^p (\M; \Z_2)=sys_s (\M; \Z_2)= \Theta(n),
\end{align}
 and Betti number of size
\be\label{eq:Betti_FH}
b_p  \equiv dim(H_p(\M; \Z_2)) = b_s = \Theta(n), 
\ee
where the correspondence between the quantities labeled by $p$ and $s$ is due to Poincar\'e duality.  The above relation are due to the fact that logical-$Z$ and $X$ information of the skeleton quantum cocde $\bar{\C}$ is encoded into the $p$-cycle and $s$-cycle ($p$-cocycle) respectively.    
Furthermore, $b_j$$=$$0$ for any $j \neq 0,1,2, p,s,r, r-1, r-2$, i.e., there are no nontrivial cycles or  cocycles of dimensions other than $p$, $s$, $0$ and $r$ except the possible existence of spurious 1-(co)cycles and 2-(co)cycles and their dual $(r-1)$-(co)cycles and $(r-2)$-(co)cycles.    These spurious (co)cycles can potentially lead to short (co)systoles in dimensions $1,2,r-1,r-2$, i.e., scales less than $\Theta(n)$ (or even $O(1)$ in some cases). Finally, we have $b_0=b_r=1$ since the manifold just have a single connected component.  

For a set of six positive integers $I = \{p_i, s_i\}_{i=0}^2$ (with the above conditions on $p_i$ and $s_i$), we can define the product manifold 
\begin{align}
\widetilde{\M}(I) =& FH(p_0,s_0) \times FH(p_1,s_1) \times FH(p_2,s_2)   \cr
\equiv & \M \times \M' \times \M''.
\end{align}
The product manifold still has bounded geometry since the vertex degree is only increased by a constant and property 1 is hence satisfied. We let each constitutent manifold $FH(p_i, s_i)$ having $\Theta(n)$ vertices, so the product manifold will have $N=\Theta(n^3)$ vertices.   Since the vertex degree are bounded, every vertex will only be adjacent to bounded number of $3q$-simplices (top dimension).  We hence have the combinatorial volume to be 
\be
vol(\widetilde{M})= |\widetilde{M}|_{3q}= \Theta(N)=\Theta(n^3),   
\ee
where $|\widetilde{M}|_{3q}$ counts the total number of $3q$-simplices on the triangulation of $\widetilde{M}$.  Therefore, property 2 is satisfied.

We now impose some conditions on these integers. First for some fixed integer $q$, we require that for each $i \in \{0,1,2\}$
\be
p_{i \text{ mod } 3} + s_{(i+1) \text{ mod } 3} = q.
\ee
This ensures that $\widetilde{\M}$ is a $3q$-dimensional manifold, and by the Kunneth theorem one has
\begin{align}
 & H_q(\widetilde{\M}; \Z_2) \cr
 \cong & \bigoplus_{l+m+t=q} H_{l}(\M; \Z_2)  \otimes H_{m}(\M'; \Z_2) \otimes H_{t}(\M''; \Z_2), \cr
\end{align}    
and hence the $q^\text{th}$ Betti number of $\widetilde{\M}$: 
\begin{align}
\tilde{b}_q \equiv& \dim(H_q(\widetilde{\M}; \Z_2)) = \sum_{l+m+t=q} b_l \cdot b'_m \cdot b''_t \cr
 \ge & \sum_{i=0}  b_{p_{i \text{ mod } 3}} \cdot b'_{s_{(i+1) \text{ mod } 3}} \cdot b''_0 = \Omega(n) \cdot \Omega(n) \cdot 1  \cr
             =& \Omega(n^2) = \Omega(N^{\frac{2}{3}}),
\end{align}
which leads to property 3.
Note that we have used Eq.~\eqref{eq:Betti_FH} in the second line of the above equation.

In fact, the Kunneth theorem can be used to completely describe the structure of $H^*(\widetilde{\M}; \Z_2)$. The relevant non-trivial triple cup products can be defined as follows. Let $\as^{p_i}$ be a basis $p_i$-cocycle in $H^{p_i}(FH(p_i, s_i))$ and let $\bs^{s_i}  \in H^{s_i}(FH(p_i, s_i))$ be a basis $s_i$-cocycle so that $\as^{p_i} \cup \bs^{s_i} \neq 0$, which exists by Poincar\'e duality. Let $\cs^0$ be the unique 0-cocycle of $FH(p_i, s_i)$. We hence have the non-trivial trip cup product $\as^{p_i} \cup \bs^{s_i} \cup \cs^0 \neq 0$. Let the Kunneth map be 
$$K: H^*(\M; \Z_2) \otimes H^*(\M'; \Z_2) \otimes H^*(\M''; \Z_2) \to H^*(\widetilde{\M}, \Z_2).$$
\noindent We then define three $q$-cocycles in $H^q(\widetilde{\M}, \Z_2)$ with the tensor product of cocycles from each constituent manifold as
\begin{align}\label{eq:_three_q_class}
\alpha^q =& \as^{p_0} \otimes \bs^{s_1} \otimes \cs''^0    \cr
\beta^q =&  \cs^0 \otimes {\as}^{p_1}  \otimes \bs^{s_2}  \cr
\gamma^q =&  \bs^{s_0} \otimes \cs'^0 \otimes  {\as}^{p_2}.
\end{align}

\noindent By the Kunneth theorem we have
\begin{align}
& \alpha^q  \cup \beta^q  \cup \gamma^q  \cr
=& (\as^{p_0} \cup \cs^0 \cup \bs^{s_0} )\otimes
 (\bs^{s_1}  \cup  {\as}^{p_1} \cup  \bs'^0) \otimes (\cs''^0 \cup \bs^{s_2}  \cup {\as}^{p_2}) \cr
 \neq & 0.
\end{align}

Note that due to Poincar\'e duality, for a given cocycle $\as^{p_i}$ in a cocycle basis $\{\as^{p_i}\}$, there is only a unique dual cocycle $\bs^{s_i} = {\as^{*}}^{p_i} $ that has a non-trivial cup product with $\as^{p_i}$, i.e., $\as^{p_i} \cup {\as^{*}}^{p_i}$$\neq$$0$.  Since we were free to choose the $\as^{p_i}$'s (with $\Theta(n)$ choices according to Eq.~\eqref{eq:Betti_FH}) which then fixes $\bs^{s_i}$'s, this gives $\Theta(n^3)=\Theta(N)$ non-trivial triple cup products and triple intersection points, satisfying property 5.

Finally, we need to show that the $q$- and $2q$-cosystoles (equivalently $2q$- and $q$-systoles) of $\widetilde{\M}$ have lower bound $\Omega(n^2)=\Omega(N^{\frac{2}{3}})$, i.e., property 4 of Theorem \ref{theorem:3q-manifold}. Due to the additional non-trivial (co)cycles [including the spurious (co)cycles] with dimension $0,1,2, r_i, r_i-1, r_i-2$, which do not come from the skeleton quantum code $\bar{\C}$ and can have short (co)systoles, the combination of these short (co)cycles themselves or together with the large (co)cycles with dimension $p_i$ and $s_i$ can lead to short (co)cycles of dimension $q$ or $2q$ with (co)cystoles less than $O(n^2)$ which would violate property 3 of Theorem \ref{theorem:3q-manifold}. The intuition is that when $q$ becomes large enough, the (co)homology group of most dimensions are just trivial, except the dimensions mentioned above.  Therefore, we expect that there exists some gaps in the combined short (co)cycle dimensions so that $q$ and $2q$ could stay in the gaps. 

The idea to solve this problem is that if we impose some more conditions on the set $I$, then we can control which types of classes appear in $H^q(\widetilde{\M})$ and $H^{2q}(\widetilde{\M})$. Note that any cohomology class which involves a product of at least two $p_i$ or $s_i$ dimensional terms coming from different $FH(p_i, s_i)$'s, such as $\alpha^q, \beta^q$ and $\gamma^q$ in Eq.~\eqref{eq:_three_q_class}, has an $\Omega(n^2)$ lower bound in size. So for the given set of parameters $I = \{p_i, s_i\}_{i=0}^2$ which completely determines $\widetilde{M}(I)$, the ``bad" dimensions which can have cosystoles of size smaller than $O(n^2)$ fall into the following two sets $B_1 \cup  B_2$:
\begin{widetext}
\begin{align}
B_1 = & \big\{p_i +x +y, \ s_i +x +y \ \big| \ i \in \{0,1,2\}, \ x  \in \{0,1,2, r_{(i+1) \text{ mod } 3}, r_{(i+1) \text{ mod } 3}-1, r_{(i+1) \text{ mod } 3}-2 \}, \cr
& y  \in \{0,1,2, r_{(i+2) \text{ mod } 3}, r_{(i+2) \text{ mod } 3}-1, r_{(i+2) \text{ mod } 3}-2 \} \big\};  \cr
B_2 = & \big\{x +y +z \ \big|  \ x  \in \{0,1,2, r_0, r_0-1, r_0-2 \},  y  \in \{0,1,2, r_1, r_1-1, r_1-2 \},  z  \in \{0,1,2, r_2, r_2-1, r_2-2 \}\big\},  \cr
\end{align}
\end{widetext}
where $r_i = p_i + s_i$.  The necessary condition for the parameter set $I$ to be valid is that the ``bad" dimensions do not contain $q$ and $2q$, i.e., $q, 2q \notin B_1 \cup B_2$.

We do an exhaustive numerical search through  all possible $I$ with increasing dimension of $q$ starting from the minimum dimension $q$=11.   The minimum dimension for a valid parameter set $I$ is $q$=31 with
\be
\non \big\{ \{p_1, s_1 \},  \{p_2, s_2 \},  \{p_3, s_3 \}  \big\} = \big\{  \{9, 16\},\{12, 22\}, \{15, 19\}  \big\}.
\ee
The set of bad dimensions $B_1 \cup B_2$ that could have short systoles/cosystoles in this case is
\begin{align}
& \{ 0,1,2,3,4,5,6,[\ ],9,10,11,12,13,14,15,16,17,18,19,20, \cr 
&21,22,23,24,25,26,27,28,29, [\ ], 32,33,34,35,36,37,38, \cr
& 39, 40,41,42, 43,44,45,46,47,48,49,50,51,52,53,54,55, \cr
& 56, 57,58,59,60,61, [\ ], 64,65,66,67,68,69,70,71,72,73, \cr
\non & 74, 75, [\ ], 77,78,79,80,81,82,83, 84,85,86,87,88,89,90, \cr
& 91, 92, 93 \},
\end{align}
which have some small gaps in the bad dimensions:  7-8, 30-31, 62-63, 76.  This allows $q$=31 and $2q$=62 to reside in the gaps. Note that the three constituent manifolds in the homological product construction has different dimensions, with $r_1=25$, and $r_2=r_3=34$ respectively.

The next valid dimension is $q=32$ with the parameters
\be
\non \big\{ \{p_1, s_1 \},  \{p_2, s_2 \},  \{p_3, s_3 \}  \big\} = \big\{  \{9, 17\},\{12, 23\}, \{15, 20\}  \big\}.
\ee
The next few valid dimensions are $q$$=$$33, 34, 35, 36, 37, 38, \cdots$,  and most of the dimensions (possibly all) in $q \ge 31$ will admit one or more valid parameter sets, which have the properties 
\begin{align}
sys_q(\widetilde{\M}; \Z_2) = sys^{2q}(\widetilde{\M}; \Z_2) =\Omega(n^{2})=\Omega(N^{\frac{2}{3}}), \cr sys^{q}(\widetilde{\M}; \Z_2) = sys_{2q}(\widetilde{\M}; \Z_2) =  \Omega(n^{2}) = \Omega(N^{\frac{2}{3}}),
\end{align}
satisfying property 3 in Theorem \ref{theorem:3q-manifold}.   Note that when keeping increasing the dimension $q$,  there are more and larger gaps in the ``bad dimensions" $B_1 \cup B_2$ and it will be easier to find qualified manifolds.

According to Eq.~\eqref{eq:combinatorial_SR},  we can evaluate the $\ZZ_2$-systolic ratio as 
\begin{align}
\ZZ_2\text{-}(q,2q)\text{-}\text{SR}(\widetilde{M}) =& \frac{sys_q(\widetilde{\M}) \cdot sys_{2q}(\widetilde{\M}) }{vol(\M)} \cr
=& \Omega(N^{\frac{1}{3}}) =\Omega(vol(\M)^{\frac{1}{3}}),
\end{align}
which shows a power-law $\ZZ_2$-$(q,2q)$-systolic freedom according to Definition \ref{def:power-law}.
\end{proof}

Due to the above theorem, we can hence reach the following corollary:
\begin{corollary}\label{corollary:subspace}
For some value $q\ge 31$, there exist a family of triple good (subspace)  codes $\tilde{\C}$ defined on the triangulation $\L$ of a $3q$-dimensional product manifold $\widetilde{\M}$ with dimension $K=\Theta(N^{\frac{2}{3}})$, code distance $d=\Omega(N^{\frac{2}{3}})$ and constant stabilizer weight $w=O(1)$,  such that a constant-depth circuit implementing the cohomology operation of a triple cup product gives rise to $\Theta(N)$ logical CCZ gates on three copies of $\tilde{\C}$. 
\end{corollary}
\begin{proof}
According to Lemma \ref{lemma_gate_2}, we choose the following constant-depth circuit: 
\be
U = (-1)^{\int_{\L} \hat{a}_{(1)}^{q} \cup \hat{a}_{(2)}^{q} \cup \hat{a}_{(3)}^{q}},
\ee
where $\hat{a}^q_{(i)}$ represents the operator-valued $q$-cochain (gauge fields) in the $i^\text{th}$ copy of thickened homological product code quantum code.

According to Eq.~\eqref{eq:logical_CCZ_higher}, a logical CCZ is applied to three logical qubits labeled by the basis cocycles $\alpha^q$, $\beta^q$ and $\gamma^q$ if there is a non-trivial cup product evaluation: 
\be
\int_{\L} \alpha^{q} \cup \beta^{q} \cup \gamma^q=1.
\ee
Since there are $\Theta(N)$ such non-trivial triple cup products according to Theorem \ref{theorem:3q-manifold},  we get in total $\Theta(N)$ logical CCZ's applied simultaneously. 

\end{proof}
 
Note that in the above manifold construction, one could add further constraint in ``bad" dimensions to make sure the $2q$-systole is of size $\Omega(N)$, which will leads to an $\Omega(N)$ $X$-distance in the above subspace code construction just the same as the situation in the subsystem code.  We leave that improvement into future work or an updated version of the current work. 

Finally, we note that although the constructed family of $3q$-manifolds in the above (including the 33-manifold constructed for subsystem codes) have varying topological dimensions at small scale which does affect the structures of the triangulation,  they have the same coarse dimension of  6D, since the manifolds constructed by Freedman and Hastings $FH(p_i, s_i)$ in Ref.~\cite{freedman:2020_manifold_from_code} used in the triple product  are coarsely 2D according
to Ref.~\cite{portnoy2023local}.  Therefore, this sequence of $3q$-manifolds look similar at large scale. 

\section{Logical gate structure and the magic state fountain}\label{sec:fountain}

Here, we analyze the detailed structure of the collective logical CCZ gate implemented by the constant-depth circuit $U$, and its application to the \textit{magic state fountain}.   

For concreteness, we focus on the subsystem code example based on the 33-manifold in Theorem \ref{theorem:subsystem} of Sec.~\ref{sec:subsystem}. The logical gate structure is completely isomorphic for the CSS subspace code construction in Theorem \ref{theorem:3q-manifold} and Corollary \ref{corollary:subspace} in Sec.~\ref{sec:subspace}. 

\subsection{Logical gate and triple-intersection structure}

\begin{figure*}[hbt]
\includegraphics[width=2\columnwidth]{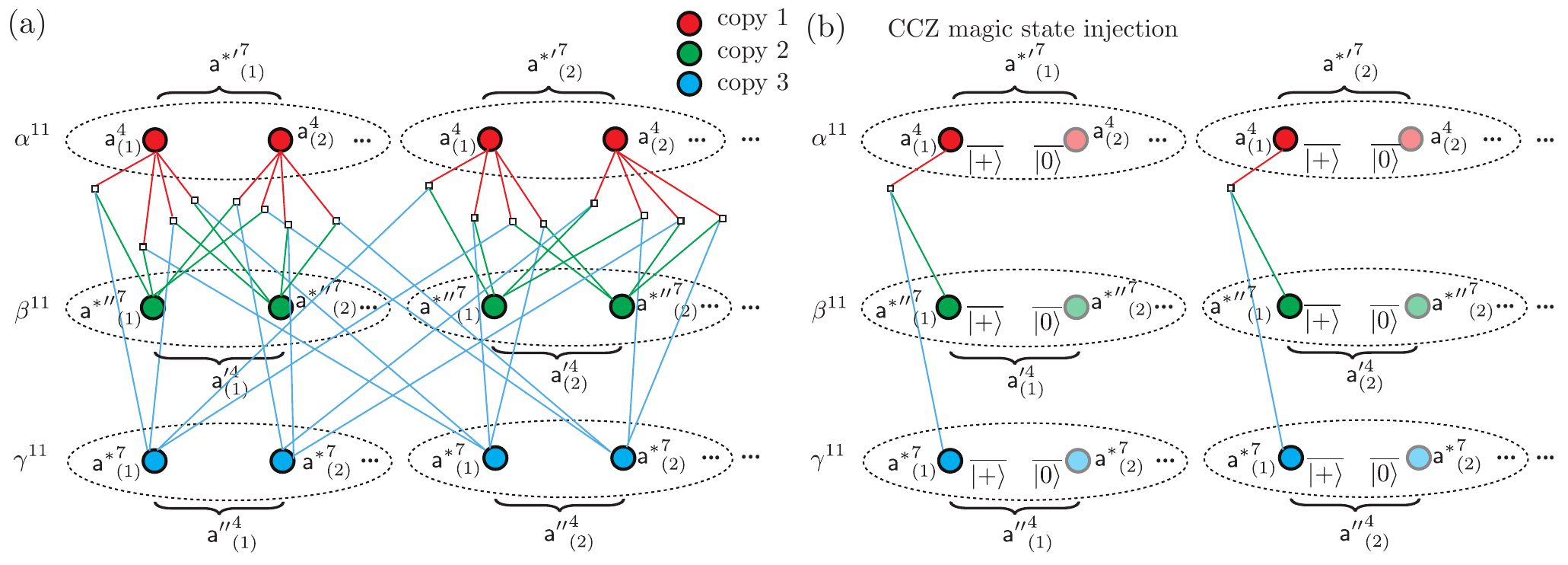}
\caption{(a) The logical gate structure represented by the interaction hypergrpah. Each vertex (circle) with red, green or blue color represent a logical qubit in the code copy 1, 2, or 3 respectively. Each hyperedge consisting of three edges and a square junction in the center represents a logical CCZ acting on the three qubits it connects to.  We see that each red logical qubit participates in logical CCZ's with $\Theta(N^{1/3}) $ green logical qubits and $\Theta(N^{1/3}) $ blue logical qubits. There are in total $\Theta(N)$ logical CCZ's. (b) To prepare logical CCZ magic states, one can initialize $\Theta(N^{1/3})$ red, green and blue logical qubits in the logical plus state $\lo{\ket{+}}$ and others in the $\lo{\ket{0}}$ state (dim circles).   In this way, one effectively turns off the hyperedges coupled to $\lo{\ket{0}}$. There are hence $\Theta(N^{1/3})$ non-overlapping hyperedges and logical CCZ's. This allows the preparation of $\Theta(N^{1/3})$ non-overlapping logical CCZ magic states.  }\label{fig:interaction-hypergraph_triple} 
\end{figure*}

The logical CCZ structure is completely determined by the triple instersection structure in the underlying manifold, which can be described by an interaction hypergraph in Fig.~\ref{fig:interaction-hypergraph_triple}(a), as has been introduced previously in Ref.~\cite{zhu2023non} and \cite{zhu2025topological}.  The vertices (circles) represent logical qubits labeled by the basis cocycles $\alpha^{11}$,  $\beta^{11}$, and $\gamma^{11}$ respectively.  The vertices in the three copies of quantum codes are repsented by red, green and blue respectively.  Each hyperedges composed of three edges (red, green and blue) meeting at a junction (square) connects to three  logical qubits in the three copies of codes. 

Recall that, according to Eq.~\eqref{eq:cocycle_class1},  one can decompose each basis cocycle  into a tensor product of three component cocycles in each constituent manifold, e.g., $\alpha^{11} = \as^4 \otimes {{\as^{*}}'}^7 \otimes \cs''^0$. Since the 0-cocycle $\cs''^0$ is unique, we can efficetively label $\alpha^{11}$ with two labels $\as^4_{(i)}$ and ${{\as^{*}}'}^7_{(i)}$, where $i$ indexes different elements in the cocycle basis.  In Fig.~\ref{fig:interaction-hypergraph_triple}, we use a dashed circle to represent a group of logical qubits with the same ${{\as^{*}}'}^7_{(i)}$ label, while the other label $\as^4_{(i)}$ can have $\Theta(n)$ choices. There exists   $\Theta(n)$ such circles.  We also represent $\beta^{11}$ and $\gamma^{11}$ with two labels in a similar way.   

There exist $\Theta(n^3)=\Theta(N)$  hyperedges which represent the same number of triple intersection points.  As illustrated in Fig.~\ref{fig:interaction-hypergraph_triple}(a), each logical qubit labeled by $\alpha^{11}$ (red) interacts with $\Theta(n)$ logical qubits labeled by $\beta^{11}$ (green) and $\Theta(n)$ logical qubits labeled by $\gamma^{11}$ (blue).  The detailed structure is completely determined by Poincar\'e duality. For example, the logical qubits with the label $\as^4_{(i)}$ only interact with those labeled by its dual cocycle ${\as^*}^7_{(i)}$; those with the label ${{\as^{*}}'}^7_{(i)}$ only interact with those labeled by its dual cocycle $\as'^4_{(i)}$.

\subsection{Application to the magic state fountain}

Now we apply the collective logical gates for the purpose of single-shot injection of 
a vast number of magic states, dubbed as the \textit{magic state fountain} \cite{zhu2023non, zhu2025topological}. 

We denote the logical CCZ magic states on three logical qubits with cocycle label $\alpha^{11}$, $\beta^{11}$ and $\gamma^{11}$  as
\be
\overline{\ket{\text{CCZ}}}_{\alpha, \beta, \gamma} := \lo{\text{CCZ}}_{\alpha, \beta, \gamma} \lo{\ket{+++}}_{\alpha, \beta, \gamma},
\ee
where the three basis cocycles satisfy the triple intersection condition $\int_{\M^{11}} \alpha^{11} \cup \beta^{11} \cup \gamma^{11}=1$.   In order to inject the logical CCZ magic states supported on non-overlapping set of logical qubits,  we can initialize only $\Theta(n)= \Theta(N^{\frac{1}{3}})$ logical qubits in the $\lo{\ket{+}}$ state, while keeping all the other logical in the $\lo{\ket{0}}$ state. This effectively turn off all the hyperedges (triple intersections) that connects to logical qubit in the $\lo{\ket{0}}$ state and preserve $\Theta(n)= \Theta(N^{\frac{1}{3}})$ of them.  When  applying the constant depth circuit $U$,  we have then injected  $\Theta(N^{\frac{1}{3}})$  logical magic states fault-tolerantly with a $\Omega(N^{\frac{2}{3}})$ distance  ($\Theta(N)$ $X$-distance) in a single shot.   In contrast the topological color code can only inject $1$ logical magic state with  distance $O(N^{\frac{1}{3}})$.

\section{Discussion and Outlook}
In summary, the current paper has achieved the first family of quantum codes that admit transverse non-Clifford gates with distance $\Omega(N^{\frac{2}{3}})$ ($X$-distance $\Theta(N)$), overcoming the $\sqrt{N}$ distance barrier. The dimension of the code is $\Theta(N^{\frac{2}{3}})$ and can be used to fault-tolerantly prepare $\Theta(N^\frac{1}{3})$ logical qubits in a single shot.  This also leads to the discovery of a family of exotic coarsely 6D manifolds with power-law $\ZZ_2$ systolic freedom and the coexistence of triple intersections between cycles with large systoles.

The subsystem code constructed in this work has dimension 33. Nevertheless, since the subsystem code idea considers the logical qubits corresponding to short cycles as gauge qubits, one can try to further lower the dimension of the construction.   The minimal dimension of 11 in the Freedman-Hastings mapping from quantum codes to manifolds in Ref.~\cite{freedman:2020_manifold_from_code} is mainly to avoid short spurious cycles in lower dimensions. Indeed, Ref.~\cite{zhu2025topological} has constructed a 4-dimensional manifold from classical codes and takes homological product of it to construct the susbystem qLDPC codes with constant rate and power-law distance.  Similarly, one could also try to lower the dimension of the manifold mapped from the quantum code to 5D by putting the qubits on 2-simplices. By taking the homological product of three 5-manifolds will lead to a subsystem code defined on a 15D triangulated manifold. We leave this exploration to future work or an updated version of the current work. We also note that although the dimension of the manifold is lowered, it still remains coarsely 6D since the geometry at large scale has not changed.  

Although the current work focuses on asymptotic scaling, the technique of constructing large-distance code is quite general and is applicable to small or midsized codes.   For example, one could choose the input quantum code for the manifold construction to be some finite-size codes with large distance such as the bivariate-bicycle code \cite{Bravyi:2024wc}, and the constructed product code will also inherit the large distance in the input code.

Finally, another more important future direction will be further pushing the distance to linear and achieving asymptotically good qLDPC codes with transversal non-Clifford gates.

\vspace{0.5in}

\noindent{\it Acknowledgements} --- We thank Michael Freedman and Elia Portnoy for insightful discussion and the original idea of the triple product construction with three good qLDPC codes.  We also thank Andrew Cross, Shehryar Sikander, Ben Brown, Po-Shen Hsin, Ryohei Kobayashi, and Maissam Barkeshli for the previous collaboration on related projects. G.Z. is supported by the U.S. Department of Energy, Office of Science, National Quantum Information Science Research Centers, Co-design Center for Quantum Advantage (C2QA) under contract number DE-SC0012704.

\begin{appendix}
\section{Proof of Lemma \ref{lemma:subsystem}}\label{app:proof}
\begin{proof}
Due to the intersection pairing $|\alpha_i \cap \beta^*_{k-i}| = \delta_{\alpha,\beta}$, a logical $Z$ operator supported on the $i$-cycle $\alpha_i$ anticommutes solely with the logical $X$ operator supported on its dual $(k{-}i)$-cycle $\alpha^*_{k-i}$. This anticommutation relation is expressed as:
\begin{equation}
\lo{Z}_{\alpha_i} \lo{X}_{\alpha^*_{k-i}} = -\lo{X}_{\alpha^*_{k-i}} \lo{Z}_{\alpha_i}.
\end{equation}
As a result, an $X$-type error must wrap around the dual cycle $\alpha^*_{k-i}$ in order to flip the eigenvalue of the logical operator $\lo{Z}_{\alpha_i}$. Conversely, a $Z$-type error supported on the cycle $\alpha_i$ is required to flip the eigenvalue of the logical operator $\lo{X}_{\alpha^*_{k-i}}$.

It follows that the $Z$-distance of the subsystem code is determined by the shortest representative among the set of basis $i$-cycles $\{\alpha_i\}$, that is,
\begin{equation}
d_Z = \min\{|\alpha_i|\},
\end{equation}
and similarly, the $X$-distance is given by the minimal size among the dual basis $(k{-}i)$-cycles $\{\beta^*_{k-i}\}$,
\begin{equation}
d_X = \min\{|\beta^*_{k-i}|\}.
\end{equation}
Therefore, the overall code distance of the subsystem code is the minimum of the two:
\begin{equation}
d = \min(d_X, d_Z) = \min\left(\min\{|\alpha_i|\}, \min\{|\beta^*_{k-i}|\}\right).
\end{equation}
\end{proof}

\end{appendix}

\bibliography{mybib_merge.bib}

\end{document}